%% file: main.tex
\documentclass{article}

\usepackage[nonatbib,final]{neurips_2022}

\usepackage[utf8]{inputenc} % allow utf-8 input
\usepackage[T1]{fontenc}    % use 8-bit T1 fonts
\usepackage{hyperref}       % hyperlinks
\usepackage{url}            % simple URL typesetting
\usepackage{booktabs}       % professional-quality tables
\usepackage{amsfonts}       % blackboard math symbols
\usepackage{nicefrac}       % compact symbols for 1/2, etc.
\usepackage{microtype}      % microtypography
\usepackage{xcolor}         % colors

\usepackage{amsmath}
\usepackage{amssymb}
\usepackage{mathtools}
\usepackage{amsthm}
\usepackage{nicefrac}
\usepackage{comment}
\usepackage{algorithm}
\usepackage{algorithmic}

\usepackage{microtype}
\usepackage{graphicx}
\usepackage{subfigure}
\usepackage{booktabs} % for professional tables
\usepackage{tabularx}
\usepackage{wrapfig}

\usepackage[capitalize,noabbrev]{cleveref}

\usepackage{pgfplots}
\pgfplotsset{compat=1.16}
\usepgfplotslibrary{fillbetween, groupplots}

\definecolor{cycle1}{RGB}{0,201,221}
\definecolor{cycle2}{RGB}{193,27,131}
\definecolor{cycle3}{RGB}{0,84,0}
\definecolor{cycle4}{RGB}{239,142,43}
\definecolor{cycle5}{RGB}{0,131,221}
\definecolor{cycle6}{RGB}{0,190,128}
\definecolor{cycle7}{RGB}{166,19,14} 
\definecolor{cycle8}{RGB}{233,133,243}
\definecolor{cycle9}{RGB}{133,126,0}
\definecolor{cycle10}{RGB}{68,61,172}

\newtheorem{theorem}{Theorem}[section]

\newtheorem{claim}[theorem]{Claim}
\newtheorem{observation}[theorem]{Observation}
\newtheorem{lemma}[theorem]{Lemma}
\newtheorem{corollary}[theorem]{Corollary}
\newtheorem{definition}[theorem]{Definition}

\newtheorem{fact}[theorem]{Fact}
\newtheorem{remark}[theorem]{Remark}

\usepackage{xspace}
\usepackage{hyperref}
\usepackage{balance}

\usepackage[textsize=tiny]{todonotes}

\newcommand{\p}{\mathbf{p}}
\newcommand{\lap}{\mathrm{Lap}}

\newcommand{\tp}{\textit{type}}
\newcommand{\joint}{\textit{joint}\xspace}
\newcommand{\njoint}{\textit{non-joint}\xspace}

\newcommand{\vol}{\mathrm{Vol}}
\renewcommand{\hat}{\widehat}

\newif\ifappendix
\appendixtrue

\newcommand{\refappendix}[1]{\ifappendix
 Appendix~\ref{#1}\xspace
\else
 the supplementary material\xspace
\fi}

\newif\ifquestionnaire
\questionnairefalse
\title{Differentially Private Graph Learning via Sensitivity-Bounded Personalized PageRank}

\newcommand*\samethanks[1][\value{footnote}]{\footnotemark[#1]}
\author{Alessandro Epasto\thanks{Google Research. \{aepasto, mirrokni, bperozzi, tsitsulin, peilinz\}@google.com} \And Vahab Mirrokni\samethanks \And Bryan Perozzi\samethanks \AND
Anton Tsitsulin\samethanks \And Peilin Zhong\samethanks}

\allowdisplaybreaks

 \hypersetup{
           breaklinks=true,   % splits links across lines
           colorlinks=true,   % displays links as colored text instead of blocks
           pdfusetitle=true,  % \title and \author values into pdf metadata
        }

\begin{document}

\maketitle

\begin{abstract}
\input{abs}
\end{abstract}

\input{intro}
\input{preli}
\input{tech}
\input{exp}
\input{conclusion}
\bibliographystyle{plain}
\bibliography{ref}

%%%%%%%%%%%%%%%%%%%%%%%%%%%%%%%%%%%%%%%%%%%%%%%%%%%%%%%%%%%%

%%%%%%%%%%%%%%%%%%%%%%%%%%%%%%%%%%%%%%%%%%%%%%%%%%%%%%%%%%%%

\ifappendix
\newpage
\appendix
\onecolumn
\input{app}
\fi 

\end{document}

%% file: abs.tex
Personalized PageRank (PPR) is a fundamental tool in unsupervised learning of graph representations such as node ranking, labeling, and graph embedding. However, while data privacy is one of the most important recent concerns, existing PPR algorithms are not designed to protect user privacy. PPR is highly sensitive to the input graph edges: the difference of only one edge may cause a large change in the PPR vector, potentially leaking private user data.

In this work, we propose an algorithm which outputs an approximate PPR and has provably bounded sensitivity to input edges. In addition, we prove that our algorithm achieves  similar accuracy to non-private algorithms when the input graph has large degrees. Our sensitivity-bounded PPR directly implies private algorithms for several tools of graph learning, such as, differentially private (DP) PPR ranking, DP node classification, and DP node embedding. To complement our theoretical analysis, we also empirically verify the practical performances of our algorithms.

%% file: intro.tex
\vspace{-0.2in}
\section{Introduction}
\vspace{-0.1in}

Personalized PageRank (PPR)~\cite{haveliwala2003topic}, has been a workhorse of graph mining and learning for the past twenty years. Given a graph $G$, and a source node $s$, the PPR vector of node $s$ defines a notion of proximity of the other nodes in the graph to it. More precisely, the proximity of $s$ to $v$, is defined by the probability that a biased random walk starting in $s$, visits node $v$. 

This elegant variation of the celebrated PageRank algorithm~\cite{page1999pagerank} has found widespread use in different applications, including web search~\cite{cho2007rankmass}, link prediction~\cite{liben2007link}, network analysis~\cite{ivan2011web,gleich2015pagerank}, graph clustering~\cite{andersen2007using}, natural language processing~\cite{scozzafava2020personalized}, spam and fake account detection~\cite{alvisi13sok,AndersenBCHJMT08}.
More recently,  PPR has also been used in  
graph neural networks~\cite{klicpera2018predict} and graph representation learning~\cite{postuavaru2020instantembedding} including to speed up the computation of graph-based learning algorithms~\cite{bojchevski2019pagerank,bojchevski2020scaling}. 

Despite the widespread use of PPR, and the extensive algorithmic literature dedicated to its efficient approximation~\cite{bahmani2011fast,andersen2007using,fogaras2005towards,hou2021massively}, to the best of our knowledge no prior work has attempted to compute PPR vectors in a privacy-preserving manner. 
In this work, we address this limitation by defining the first approximate method for PPR computation with differential privacy~\cite{dwork2006calibrating}.
We focus on a standard notion of differential privacy (DP) for graphs known as edge-level DP. In this notion, two unweighted, undirected graphs $G=(V,E)$ and $G'=(V,E')$, are deemed neighbors if they differ only in the presence of a single edge~\cite{marc2021differentially,eliavs2020differentially}. An algorithm $\mathcal{A}$ is then said to be edge-level $\epsilon$-differentially private if the difference in the probability of observing any particular outcome from the algorithm when run on $G$ vs $G'$ is bounded:
$\Pr[\mathcal{A}(G) \in S] \leq e^{\epsilon}\cdot \Pr[\mathcal{A}(G') \in S].$

Edge-level DP guarantees a strong notion of plausible deniability for the existence of an edge in the graph. This is especially critical in graph-based learning applications, where nodes correspond to humans, and edges depend on personal relationships, which can be highly private and sensitive. Achieving edge-level DP ensures that an attacker observing the output of the algorithm is information-theoretically bounded in their ability to uncover any specific user pair connection.

On top of this notion, we also explore a popular variation of differential privacy  used in personalization~\cite{jain2021differentially} that follows the concept of joint differential privacy~\cite{kearns2014mechanism}. 
More precisely we provide Personalized PageRank algorithms that are joint, edge-level differentially-private with respect to the neighborhood of the source node $s$. This means that we can provide the user corresponding to $s$, with an approximate Personalized PageRank of $s$ that depends on the edges incident to $s$ but that protects the information on edges of the rest of the graph.
This notion is especially relevant in the context of personalization of results in social networks using PPR, where user data can be safely used 
to provide an output to the user, but must be protected from leaking to others.

\vspace{-0.1in}
\subsection{Our results and outline of the paper}
\vspace{-0.1in}

Differential privacy forces an algorithm to be insensitive to changes of an edge in the graph. This makes the design of DP PPR algorithms especially challenging as a single edge removal or addition may result in dramatically different PPR vectors, thus potentially exposing private user data. 
Our first contribution is to propose an algorithm that approximates the PPR vector with a provably bounded sensitivity to edge changes. 
This technical contribution directly leads to the design of edge-level DP (and joint edge-level DP) algorithms for computing approximate PPR vectors. We believe that this technique may have broader applications in the design of DP graph algorithms.

From a theoretical standpoint, 
we prove that our private algorithms achieve similar accuracy as non-private approximation algorithms when the input graph has a large enough minimum degree (while the privacy guarantee holds for all graphs of any degree). This dependency on the degree is theoretically justified as we show non-trivial approximation requires large enough degrees.

The main ingredient of our DP algorithms is a novel (non-private) sensitivity-bounded approximate PPR algorithm (Algorithm~\ref{alg:capped_push_flow}).
For any input parameter $0<\sigma<1$, the  sensitivity of the output of the algorithm in the (joint) edge-level DP case is always upper bounded by $\sigma$. In addition, we show that the algorithm has an $O(\sigma)$ additive error to the ground truth PPR when the minimum degree of the graph is $\Omega(1/\sigma)$ in the DP case (resp. $\Omega(\sqrt{1/\sigma})$ in the joint-DP case).
We show that this requirement on the minimum degree for approximation guarantees is almost tight due to hard instances that we present in~\refappendix{sec:tight_bound_of_sensitivity_of_ppr}.
We then use our sensitivity-bounded algorithm to obtain an edge-level DP  (resp. joint edge-level DP) algorithm that, for a graph of minimum degree at least $D$, has $O(1/D)$ additive error (resp. $O(1/D^2)$ error).
Then, we focus on applications of differentially-private PPR including computing 
graph embeddings. We provide provably edge-level DP and joint edge-level DP graph embedding algorithm in Section~\ref{sec:DP_instant_embedding}.
Finally, in Section~\ref{sec:exp}, we empirically evaluate the performance of our differentially private PPR rankings, as well as that of the embedding methods we design, in several down-stream graph-learning tasks such as node ranking and classification.

To the best of our knowledge, our paper presents the first algorithm with theoretical guarantees for differentially private PPR. This contributes to the still quite short list of private graph algorithms with provable approximation guarantees that have been developed so far~\cite{marc2021differentially,eliavs2020differentially,fichtenberger2021differentially,zhang2021differentially,upadhyay2021differentially}, 

%% file: preli.tex
\vspace{-0.2in}
\section{Preliminaries}
\vspace{-0.1in}
We consider an undirected and unweighted graph $G$ with node set $V=\{v_1,v_2,\cdots,v_n\}$ and edge set $E$. 
Let $A$ be the adjacency matrix where $A_{i,j}=1$ indicates an edge between $v_i$ and $v_j$, and $A_{i,j}=0$ otherwise.
Let $\Lambda$ be the diagonal matrix where $\Lambda_{i,i}=d(v_i)$ denotes the degree of $v_i$.
We use $\mathbf{0}^n$ to denote an $n$-dimensional all-zero vector and use $e_i\in\mathbb{R}^n$ to denote the one-hot vector where the $i$-th entry is $1$ and other entries are $0$.
Let $[n]$ denote the set $\{1,2,\cdots, n\}$.
When there is no ambiguity, we sometimes abuse the notation between $[n]$ and $V$, i.e., using $v\in[n]$ to denote a node or $i\in V$ denotes an index between $1$ and $n$.
For $x\in\mathbb{R}^k$, we denote $\|x\|_1=\sum_{i\in[k]}|x_i|$ and $\|x\|_{\infty}=\max_{i\in [k]} |x_i|$.
We use $\lap(b)$ $(b>0)$ to denote the Laplace distribution with density function $f(x)=\frac{1}{2b}\cdot \exp(-|x|/b)$.
The cumulative distribution function of $\lap(b)$ is
$
F(x)=\left\{\begin{array}{ll} \frac{1}{2}\cdot \exp(x/b), & x<0, \\ 1-\frac{1}{2}\cdot \exp\left(-x/b\right), & x\geq 0. \end{array}\right.
$.
We will use the following fact in our analysis.
\begin{fact}\label{fac:laplace_bound}
Consider $Y$ drawn from $\lap(b)$.
For any $\delta\in(0,1)$, $\Pr[|Y|>b\ln(1/\delta)]=\delta$.
\end{fact}

\vspace{-0.1in}
\subsection{Personalized PageRank}\label{sec:preli_ppr}
\vspace{-0.1in}
Personalized PageRank (PPR) takes an input distribution of starting nodes $s\in\mathbb{R}^n$, and starts a lazy random walk with teleport probability $\alpha\in(0,1)$. 
Typically, $s=e_i$ for some $i\in[n]$, which enforces the random walk starting from the source node $v_i$.
If there is no ambiguity, we abuse the notation to denote with $s$ the source node.
The output PPR vector is the stationary distribution of the random walk.
Precisely, let $W=\frac{1}{2}(I+\Lambda^{-1}A)$ be the lazy random walk transition matrix.\footnote{This is equivalent to the standard random walk matrix up to a change in $\alpha$ (see \cite{andersen2007using}). We use the lazy walk for consistency with prior work.}
The PPR vector $\p(s)$ is defined recursively as: $\p(s)=\alpha\cdot s + (1-\alpha)\cdot \p(s)\cdot W$. 
In many applications, it is good enough to use approximate PPR vectors and their variants.
\begin{definition}[$\xi$-approximate PPR, e.g.,~\cite{andersen2007using}]\label{def:xi_approx_ppr}
For $\xi>0$, a $\xi$-approx. PPR vector for $\p(s)$ is a PPR vector $\p(s-r)$ where $r$ (residual vector) is a non-negative $n$-dimensional and $\|r\|_{\infty}\leq \xi$.
\end{definition}
\begin{definition}[$(\xi,\eta)$-approximate PPR]
For $\xi,\eta >0$, a $(\xi,\eta)$-approximate PPR vector $p$ for $\p(s)$  satisfies $\|p-\p(s-r)\|_{\infty}\leq \eta$ where $r$ is a non-negative $n$-dimensional vector and $\|r\|_{\infty}\leq \xi$.
\end{definition}
Note that $\xi$ denotes the error introduced by the residual and $\eta$ denotes the error introduced by the PPR vector itself. 
Two types of errors are well studied in the literature: for residual error $\xi$, see e.g.,~\cite{andersen2007using}; for PPR vector error $\eta$, see e.g.,~\cite{hou2021massively}.

\vspace{-0.1in}
\subsection{Differential Privacy}
\vspace{-0.1in}
We consider edge-level DP and joint edge-level DP for graph algorithms.
Given a graph $G$, we denote $\Gamma(G)$ as the set of all neighboring graphs of $G$, i.e., $\forall G'\in\Gamma(G)$, $G'$ can be obtained from $G$ by adding or removing an edge.

\begin{definition}[Edge-level DP~\cite{eliavs2020differentially} and joint edge-level DP~\cite{kearns2014mechanism}]
A randomized graph algorithm $\mathcal{A}$ is edge-level $\varepsilon$-DP if for any input graphs $G$, $G'$ satisfying $G'\in\Gamma(G)$, and for any subset $S$ of possible outputs of $\mathcal{A}$, it holds $\Pr[\mathcal{A}(G)\in S]\leq e^{\varepsilon}\cdot \Pr[\mathcal{A}(G')\in S].$
Let $V$ be the set of $n$ nodes (users).
For joint edge-level DP we assume that a family of $n$ (personalized) graph algorithms $\mathcal{A} = ( \mathcal{A}_1,\mathcal{A}_2,\cdots,\mathcal{A}_n)$ is run on the graph and the output of $\mathcal{A}_i$ is provided only to user (node) $v_i\in V$. The family $\mathcal{A}$ is joint edge-level $\varepsilon$-DP, if for any $x,y\in V$ and for any two neighboring graphs $G,G'$ that only differ on edge $(x,y)$, for any $v\not=x,y$ and for any subset $S$ of possible outputs of $\mathcal{A}_v$, it always holds $\Pr[\mathcal{A}_v(G)\in S]\leq e^{\varepsilon}\cdot \Pr[\mathcal{A}_v(G')\in S]$.\footnote{
Note this is slightly weaker than the original definition of joint DP which asks for $\forall$ subsets $S_1,S_2,\cdots,S_n$ of possible outputs of $\mathcal{A}_1,\mathcal{A}_2,\cdots,\mathcal{A}_n$ respectively, $\Pr[\mathcal{A}_{-x,-y}(G)\in S_{-x,-y}]\leq e^{\varepsilon}\cdot\Pr[\mathcal{A}_{-x,-y}(G')\in S_{-x,-y}]$, where $\mathcal{A}_{-x,-y}$ is the tuple of $n-2$ outputs of $\mathcal{A}_1,\mathcal{A}_2,\cdots,\mathcal{A}_n$ except $\mathcal{A}_x,\mathcal{A}_y$, and $S_{-x,-y}$ is the cartesian product of $S_1,S_2,\cdots,S_n$ except $S_x,S_y$. 

Our definition protects the privacy of each user, however, as we assume that output of $\mathcal{A}_i$ is available to (node) $v_i\in V$ only.}
\end{definition}
For $s\in V$, let us denote $\Gamma_s(G)$ as the set of graphs $G'\in \Gamma(G)$ satisfying that $G'$ and $G$ differ on an edge that is {\it not} incident to $s$.
It is easy to verify that the joint edge-level $\varepsilon$-DP is equivalent to asking for $\forall s\in V,$ $\forall$ subset $S$ of possible outputs of $\mathcal{A}_s$, $\forall G,G'\in\Gamma_s(G)$, $\Pr[\mathcal{A}_s(G)\in S]\leq e^{\varepsilon}\cdot \Pr[\mathcal{A}_s(G')\in S]$.
In the remainder of the paper, we will simply call edge-level $\varepsilon$-DP as $\varepsilon$-DP and joint edge-level $\varepsilon$-DP as joint $\varepsilon$-DP.

In the following, we will briefly review several other related definitions and theorems for DP.
\begin{definition}[Sensitivity~\cite{dwork2006calibrating}]\label{def:sensitivity}
Consider a function $f$ whose input is a graph and whose output is in $\mathbb{R}^k$. 
The sensitivity $S_f$ is defined as $S_f=\max_{G,G':G'\in\Gamma(G)}\|f(G)-f(G')\|_1$.
Consider a family $\mathcal{F}$ of functions $f_1,f_2,\cdots,f_n$ where each takes a graph as input and outputs a vector in $\mathbb{R}^k$.
The joint sensitivity $S_{\mathcal{F}}$ is defined as $S_{\mathcal{F}}=\max_{s\in[n],G,G'\in\Gamma_s(G)}\|f_s(G)-f_s(G')\|_1$.
\end{definition}

Moreover, when this simplifies the presentation, we will refer to $\epsilon$-DP and sensitivity as {\it non-joint} DP and {\it non-joint} sensitivity to oppose them to {\it joint} DP and {\it  joint} sensitivity.

\begin{theorem}[Laplace mechanism~\cite{dwork2006calibrating}]\label{thm:laplace_mechanism}
Consider a function $f$ whose input is a graph and whose output is in $\mathbb{R}^k$.
Suppose $f$ has sensitivity $S_f$.
Then the algorithm $\mathcal{A}(G)$ which outputs $f(G)+(Y_1,Y_2,\cdots,Y_k)$ is $\varepsilon$-DP where $Y_i$ are independent $\lap(S_f/\varepsilon)$ random variables.
Similarly, consider a family $\mathcal{F}$ of functions $f_1,f_2,\cdots,f_n$ with joint sensitivity $S_\mathcal{F}$.
Then the family of $\mathcal{A}_1,\mathcal{A}_2,\cdots,\mathcal{A}_n$ is joint $\varepsilon$-DP where $\mathcal{A}_i(G)$ outputs $f_i(G)+(Y_{i,1},Y_{i,2},\cdots,Y_{i,k})$, and $Y_{i,j}$ are independent $\lap(S_{\mathcal{F}}/\varepsilon)$ random variables.
\end{theorem}
\begin{theorem}[Composition~\cite{dwork2014algorithmic}]\label{thm:adv_comp}
Consider two algorithms $\mathcal{A}_1:\mathcal{X}\rightarrow \mathcal{Y},\mathcal{A}_2:\mathcal{Y}\times\mathcal{X}\rightarrow \mathcal{Z}$.
Suppose $\mathcal{A}_1(\cdot)$ is $\varepsilon_1$-DP, and $\mathcal{A}_2(Y,\cdot)$ is $\varepsilon_2$-DP for any given $Y\in\mathcal{Y}$.
Then the algorithm $\mathcal{A}_3:\mathcal{X}\rightarrow\mathcal{Z}$ which is defined as $\mathcal{A}_3(X)=\mathcal{A}_2(\mathcal{A}_1(X),X)$ is $(\varepsilon_1+\varepsilon_2)$-DP.
\end{theorem}

%% file: tech.tex
\vspace{-0.1in}
\section{Warm-Up: Push-Flow on Graphs with High Degrees}
\vspace{-0.1in}
As a warm-up, let us start with a standard push-flow algorithm for PPR~\cite{andersen2007using} and provide a novel analysis for bounding the sensitivity when each node has degree at least $D$. 
The non-private push-flow algorithm is described in Algorithm~\ref{alg:standard_push_flow}.\footnote{The algorithm described in Algorithm~\ref{alg:standard_push_flow} is slightly different from the original push-flow algorithm of \cite{andersen2007using}. 
 In each iteration, instead of pushing flow for the node with the largest residual, we push flows for all nodes that were visited.
 This variant gives us a better bound of sensitivity.
}
\begin{algorithm}[h]
	\footnotesize
	\begin{algorithmic}[1]\caption{\textsc{PushFlow}$(G,s,\alpha,\xi)$}\label{alg:standard_push_flow}
	\STATE \textbf{Input:} Graph $G=(V,E)$, source node $s\in V$, teleport probability $\alpha$, precision $\xi$.
	\STATE \textbf{Output:} Approximate PPR vector for $\p(s)$.
	\STATE Initialize $S^{(0)}\gets\{s\}$,$p^{(0)}\gets \mathbf{0}^n$, $r^{(0)}\gets e_s$, and
	$R\gets \lceil \ln(1/\xi)/\alpha \rceil$.
	\FOR{$i:=1\rightarrow R$}
	    \STATE Let $S^{(i)}\gets S^{(i-1)}.$ Let $p^{(i)},r^{(i)}\gets \mathbf{0}^n$.
	    \FOR{Each node $v\in S^{(i-1)}$}
	        \STATE $p^{(i)}_v\gets p^{(i-1)}_v+\alpha\cdot r^{(i-1)}_v$,  $r^{(i)}_v\gets r^{(i)}_v+(1-\alpha)/2\cdot r^{(i-1)}_v$
	        \STATE For each neighbor $u$, i.e., $(v,u)\in E$: 
	        $r^{(i)}_u\gets r^{(i)}_u+(1-\alpha)/2\cdot r^{(i-1)}_v/d(v)$,
	        $S^{(i)}\gets S^{(i)}\cup \{u\}$.
	    \ENDFOR
	\ENDFOR
	\STATE Output $p^{(R)}$.
	\end{algorithmic}
\end{algorithm}

\begin{lemma}\label{lem:correctness_of_standard_ppr}
Algorithm~\ref{alg:standard_push_flow} outputs a $\xi$-approximate PPR vector in $O(|E|\log(1/\xi)/\alpha)$ time.
\end{lemma}
The proof of Lemma~\ref{lem:correctness_of_standard_ppr} follows the analysis idea of~\cite{andersen2007using}. 
For completeness, we put the proof in \refappendix{sec:detailed_proof_of_correctness_of_standard_ppr}.
Next, we prove the sensitivity of Algorithm~\ref{alg:standard_push_flow} when every node has degree at least $D$.
\begin{theorem}[Sensitivity of \textsc{PushFlow}]\label{thm:sensitivity_standard_pushflow}
Consider two graphs $G=(V,E),G'=(V,E')$ where $G'\in\Gamma(G)$.
In addition, both $G$ and $G'$ have a minimum degree at least $D$.
Let $p,p'$ be the output of \textsc{PushFlow}$(G,s,\alpha,\xi)$ and \textsc{PushFlow}$(G',s,\alpha,\xi)$ respectively.
Then if $G'\in\Gamma_s(G)$, $\|p-p'\|_1\leq \frac{2\cdot (1-\alpha)}{\alpha\cdot D^2}$. 
Otherwise, $\|p-p'\|_1\leq \frac{2\cdot (1-\alpha)}{\alpha\cdot D}$.
\end{theorem}
\begin{proof}
Without loss of generality, suppose $G'$ has one more edge $(x,y)$ than $G$, i.e., $E'=E\cup \{(x,y)\}$.
Let $p^{(i)},r^{(i)}$ be the same as described in Algorithm~\ref{alg:standard_push_flow} when running \textsc{PushFlow}$(G,s,\alpha,\xi)$.
Similarly, let $p'^{(i)},r'^{(i)}$ be the vectors $p^{(i)},r^{(i)}$ described in Algorithm~\ref{alg:standard_push_flow} when running \textsc{PushFlow}$(G',s,\alpha,\xi)$.
Thus, our goal is to bound $\|p^{(R)}-p'^{(R)}\|_1$.
It suffices to bound $\|p^{(R)}-p'^{(R)}\|_1+\|r^{(R)}-r'^{(R)}\|_1$.
Let $d(v)$ denote the degree of $v$ in $G$ and let $d'(v)$ denote the degree of $v$ in $G'$.
Consider $i\in[R]$.
We have:
{
\small
\begin{align}
&\|p^{(i)}-p'^{(i)}\|_1+\|r^{(i)}-r'^{(i)}\|_1\notag\\
\leq&  \sum_{v \in V} \left|(p_v^{(i-1)}+\alpha\cdot r_v^{(i-1)})-(p_v'^{(i-1)}+\alpha\cdot r_v'^{(i-1)})\right|\notag\\
&+\sum_{v\in V}\frac{1-\alpha}{2}\cdot \left|r_v^{(i-1)}- r_v'^{(i-1)}\right|\notag+\sum_{v\in V}\frac{1-\alpha}{2}\cdot\left|\sum_{u:(u,v)\in E}  \frac{r_u^{(i-1)}}{d(u)}- \sum_{u':(u',v)\in E'}\frac{r_{u'}'^{(i-1)}}{d'(u')}\right|\notag\\
\leq &\|p^{(i-1)}-p'^{(i-1)}\|_1+\frac{1+\alpha}{2}\cdot \|r^{(i-1)}-r'^{(i-1)}\|_1\label{eq:difference_introduced_by_last_round}\\
&+\frac{1-\alpha}{2}\cdot\left(\sum_{v\in V}\sum_{u:(u,v)\in E}\left|\frac{r_u^{(i-1)}}{d(u)}-\frac{r_u'^{(i-1)}}{d'(u)}\right|+ \left(\frac{r_x'^{(i-1)}}{d'(x)}+\frac{r_y'^{(i-1)}}{d'(y)}\right)\right)\label{eq:difference_of_additional_edge}.
\end{align}
}
By reordering the terms, part~\eqref{eq:difference_of_additional_edge} is equal to:
{
\small
\begin{align*}
&\frac{1-\alpha}{2}\cdot \sum_{u\in V\setminus\{x,y\}}d(u)\cdot \left|\frac{r_u^{(i-1)}}{d(u)}-\frac{r_u'^{(i-1)}}{d(u)}\right|\\
+&\frac{1-\alpha}{2}\cdot\left(\left(d(x)\cdot \left|\frac{r_x^{(i-1)}}{d(x)}-\frac{r_x'^{(i-1)}}{d(x)+1}\right|+\frac{r_x'^{(i-1)}}{d(x)+1}\right)+\left(d(y)\cdot \left|\frac{r_y^{(i-1)}}{d(y)}-\frac{r_y'^{(i-1)}}{d(y)+1}\right|+\frac{r_y'^{(i-1)}}{d(y)+1}\right)\right).
\end{align*}
}
Notice that 
$
     d(x)\cdot \left|\frac{r_x^{(i-1)}}{d(x)}-\frac{r_x'^{(i-1)}}{d(x)+1}\right|+\frac{r_x'^{(i-1)}}{d(x)+1}\leq |r_x^{(i-1)}-r_x'^{(i-1)}|+\frac{2\cdot r_x'^{(i-1)}}{d(x)+1}.
$
Similar arguments hold for $y$. 
Thus, part \eqref{eq:difference_of_additional_edge} is at most $(1-\alpha)/2\cdot (\|r^{(i-1)}-r'^{(i-1)}\|_1+2\cdot(r_x'^{(i-1)}/d'(x)+r_y'^{(i-1)}/d'(y)))$.
Therefore \eqref{eq:difference_introduced_by_last_round}$+$\eqref{eq:difference_of_additional_edge}$\leq \|p^{(i-1)}-p'^{(i-1)}\|_1+\|r^{(i-1)}-r'^{(i-1)}\|_1 + (r_x'^{(i-1)}/d'(x)+r_y'^{(i-1)}/d'(y))$.
Since $d'(x),d'(y)\geq D$, we have:
$
\|p^{(i)}-p'^{(i)}\|_1+\|r^{(i)}-r'^{(i)}\|_1
\leq\|p^{(i-1)}-p'^{(i-1)}\|_1+\|r^{(i-1)}-r'^{(i-1)}\|_1
+(1-\alpha)\cdot \left(r_x'^{(i-1)}+r_y'^{(i-1)}\right)/D.
$

In~\refappendix{sec:proof_of_claims}, we show that ${r'}_x^{(i-1)},{r'}_y^{(i-1)}\leq (1-\alpha)^{i-1}$ and if in addition $s\not=x,y$, ${r'}_x^{(i-1)},{r'}_y^{(i-1)}\leq (1-\alpha)^{i-1}/D$.
Thus,
if $s\not=x,y$, we have:
$\|p^{(R)}-p'^{(R)}\|_1\leq 2\cdot (1-\alpha)/D\cdot \sum_{i=1}^R (1-\alpha)^{i-1}/D\leq \frac{2\cdot (1-\alpha)}{\alpha\cdot D^2}$.
Otherwise, 
we have 
$\|p^{(R)}-p'^{(R)}\|_1\leq 2\cdot (1-\alpha)/D\cdot \sum_{i=1}^R (1-\alpha)^{i-1}\leq \frac{2\cdot (1-\alpha)}{\alpha\cdot D}$
\end{proof}
In~\refappendix{sec:tight_bound_of_sensitivity_of_ppr}, we show that the analysis in Theorem~\ref{thm:sensitivity_standard_pushflow} is tight as the sensitivity of the ground truth PPR in graphs with minimum degree $D$ can be $\Omega(1/D)$ (or $\Omega(1/D^2)$ for joint sensitivity). 

If input graphs were always guaranteed to have minimum degree at least $D$, we could obtain a DP or a joint DP PPR algorithm by applying the Laplace mechanism on the output of Algorithm~\ref{alg:standard_push_flow} (see~\refappendix{sec:DPPushFlow}).
However, in practice the input graphs can have low degree nodes.
In this case, the sensitivity of the vanilla push-flow algorithm (Algorithm~\ref{alg:standard_push_flow}) can be very high.
In the next section, we address the question of how to modify the algorithm to ensure low sensitivity for {\it any} input graph.

\section{Push-Flow with Bounded Sensitivity in General Graphs}
In this section, we propose a variant of the push-flow algorithm of which sensitivity (resp. joint sensitivity) is always bounded by an input parameter $\sigma$.
As a result, we can apply Laplace mechanism (Theorem~\ref{thm:laplace_mechanism}) to obtain a DP PPR (resp. a joint DP PPR) algorithm for all possible general input graphs and thus the added noise is controlled by $\sigma$.
In addition, our new algorithm achieves the same approximation  as Algorithm~\ref{alg:standard_push_flow} when every node in the input graph has a relatively high degree.

We explain the intuition of our algorithm as the following.
Recall the analysis of the sensitivity of Algorithm~\ref{alg:standard_push_flow} (the proof of Theorem~\ref{thm:sensitivity_standard_pushflow}).
The reason that it may introduce a large sensitivity is because every node $x$ with residual ${r'_x}^{(i-1)}$ may push $\sim {r'_x}^{(i-1)}/d'(x)$ amount of flow along each of its incident edges.
If $d'(x)$ is small, the sensitivity introduced by an edge incident to $x$ can be very large. 
Therefore, to control the sensitivity, a natural idea is to set a threshold for each edge such that the total pushed flow along each edge can never be above the threshold.
We present our sensitivity-bounded push-flow algorithm in Algorithm~\ref{alg:capped_push_flow}.
For sake of presentation, this algorithm and the following ones have a parameter $\tp \in \{\joint, \njoint\}$ indicating whether we are working in the \joint DP case or in the \njoint DP case. 

\begin{algorithm}[h]
	\footnotesize
	\begin{algorithmic}[1]\caption{\textsc{PushFlowCap}$(G,s,\alpha,\xi,\sigma,\tp)$}\label{alg:capped_push_flow}
	\STATE \textbf{Input:} Graph $G=(V,E)$, source node $s\in V$, teleport probability $\alpha$, precision $\xi$, sensitivity parameter $\sigma$, and $\tp\in \{\joint,\njoint\}$ indicating whether joint DP sensitivity or (vanilla) DP sensitivity is required.
	\STATE \textbf{Output:} Approximate PPR vector for $\p(s)$.
	\STATE Initialize the set of nodes with positive residual $S^{(0)}\gets \{s\},$ PPR $p^{(0)}\gets \mathbf{0}^n$, residual $r^{(0)}\gets e_s,$ total pushed flow $h^{(0)}\gets\mathbf{0}^n$, number of rounds $R\gets \lceil \ln(1/\xi)/\alpha \rceil$, and thresholds $T\in\mathbb{R}^n$ such that 
	\begin{enumerate}
	    \item If $\tp=\joint$, $T_s\gets\infty$ and 
	    $\forall u\not= s,$  $T_u \gets \sigma / ((3-\alpha)\cdot (1-(1-\alpha)^R))$
	    \item Otherwise, 
	    $\forall u\in V,$  $T_u\gets \sigma / ((3-\alpha)\cdot (1-(1-\alpha)^R))$.
	\end{enumerate}
	\FOR{$i:=1\rightarrow R$}
	    \FOR{Each node $v\in S^{(i-1)}$}
	        \STATE $f_v^{(i)}\gets \min(r_v^{(i-1)},d(v)\cdot T_v-h_v^{(i-1)})$. \hfill{//Compute the flow to push for node $v$.}
	        \STATE $h_v^{(i)}\gets h_v^{(i-1)}+f_v^{(i)}$. \hfill{//Update the total pushed flow of node $v$.}
	        \STATE $p_v^{(i)}\gets p_v^{(i-1)},r_v^{(i)}\gets r_v^{(i-1)}-f_v^{(i)}$.
	    \ENDFOR
	    \STATE $S^{(i)}\gets S^{(i-1)}$.
	    \FOR{Each node $v\in S^{(i-1)}$ with $f_v^{(i)}>0$}
	        \STATE $p_v^{(i)}\gets p_v^{(i)}+\alpha\cdot f_v^{(i)},r_v^{(i)}\gets r_v^{(i)}+(1-\alpha)/2\cdot f_v^{(i)}$. \hfill{// Do actual flow push}
	        \STATE For each neighbor $u$, i.e., $(v,u)\in E:r_u^{(i)}\gets r_u^{(i)}+(1-\alpha)/2\cdot f_v^{(i)}/d(v)
	        ,S^{(i)}\gets S^{(i)}\cup \{u\}$.
	    \ENDFOR
	\ENDFOR
	\STATE Output $p^{(R)}$.
	\end{algorithmic}
\end{algorithm}

Firstly, we show that the joint/non-joint sensitivity of \textsc{PushFlowCap}$(G,s,\alpha,\xi,\sigma,\joint/\njoint)$ is indeed bounded by $\sigma$.
We will use the following observation.
\begin{observation}\label{obs:properies_capped_pushflow}
Consider $p^{(i)},h^{(i)},f^{(i)}$ in \textsc{PushFlowCap}$(G,s,\alpha,\xi,\sigma,\joint/\njoint)$. 
Then $\forall v\in V,i\in[R]$:
    (1) $h_v^{(i)}=\sum_{j=1}^i f_v^{(j)}$.
    (2) $h_v^{(i)}=\min(\sum_{j=0}^{i-1} r_v^{(j)}, d(v)\cdot T_v)$.
    (3) $p_v^{(i)}=\alpha\cdot h_v^{(i)}$.
\end{observation}
Using this observation we show the following lemma on $h^{(i)}_v$.
\begin{lemma}\label{lem:property_of_h}
Consider $h^{(i)}$ and $r^{(i)}$ in 
\textsc{PushFlowCap}$(G,s,\alpha,\xi,\sigma,\joint/\njoint)$.
$\forall v\in V,i\in[R]$, $h_v^{(i)}=\min\Big (r_v^{(0)}+\frac{1-\alpha}{2}\cdot \big (h_v^{(i-1)}+ \underset{u:(u,v)\in E}{\sum} \frac{h_u^{(i-1)}}{d(u)}\big),$ $d(v)\cdot T_v \Big)$.
\end{lemma}

\begin{theorem}[Sensitivity and joint sensitivity of \textsc{PushFlowCap}]\label{thm:sensitivity_capped_pushflow}
Consider two graphs $G=(V,E),G'=(V,E')$.
Let $p,p'$ be the output of \textsc{PushFlowCap}$(G,s,\alpha,\xi,\sigma,\tp)$ and \textsc{PushFlowCap}$(G',s,\alpha,\xi,\sigma,\tp)$ respectively.
For $\tp=\joint$ and $G'\in\Gamma_s(G)$, or $\tp=\njoint$ and $G'\in\Gamma(G)$, then $\|p-p'\|_1\leq \sigma$.
\end{theorem}
\vspace{-0.18in}
\begin{proof}
Without loss of generality, let $G'\in\Gamma(G)$ has exactly one more edge $(x,y)$ than $G$, i.e., $E'=E\cup \{(x,y)\}$.
Let $p^{(i)},h^{(i)}$ be the same as described in Algorithm~\ref{alg:capped_push_flow} when running \textsc{PushFlowCap}$(G,s,\alpha,\xi,\sigma,\tp)$.
Similarly, let $p'^{(i)},h'^{(i)}$ be the vectors $p^{(i)},h^{(i)}$ described in Algorithm~\ref{alg:capped_push_flow} when running \textsc{PushFlowCap}$(G',s,\alpha,\xi,\sigma,\tp)$.
Let $d(v)$ denote the degree of $v$ in $G$ and let $d'(v)$ denote the degree of $v$ in $G'$.
To prove the lemma, our goal is to bound $\|p^{(R)}-p'^{(R)}\|_1$.
According to Observation~\ref{obs:properies_capped_pushflow}, we have $\|p^{(R)}-p'^{(R)}\|_1=\alpha\cdot \|h^{(R)}-h'^{(R)}\|_1$.
It suffices to bound $\|h^{(R)}-h'^{(R)}\|_1$.
Notice that $\forall a_1,a_2,a_3,a_4\in\mathbb{R}$, it is easy to verify that
$|\min(a_1,a_2)-\min(a_3,a_4)|\leq  \max(|a_1-a_3|,|a_2-a_4|)$.
Consider $i\in [R]$.
According to Lemma~\ref{lem:property_of_h}, for every $v\not=x,y$, we have:
{
\small
\begin{align*}
|h_v^{(i)}-h_v'^{(i)}|
\leq &\max \left( |r_v^{(0)}-r_v'^{(0)}|+\frac{1-\alpha}{2}\left(|h_v^{(i-1)}-h_v'^{(i-1)}|+\sum_{u:(u,v)\in E}\left| \frac{h_u^{(i-1)}}{d(u)}-\frac{h_u'^{(i-1)}}{d'(u)}\right|\right),\left|(d(v)-d'(v)) T_v\right|\right)\\
=&\frac{1-\alpha}{2}\cdot \left(\left|h_v^{(i-1)}-h_v'^{(i-1)}\right|+\sum_{u:(u,v)\in E}\left| \frac{h_u^{(i-1)}}{d(u)}-\frac{h_u'^{(i-1)}}{d'(u)}\right|\right),
\end{align*}
}
where the last step follows from $r_v^{(0)}=r_v'^{(0)}$ and $d(v)=d'(v)$.
Suppose $v\in \{x,y\}$.
Let $v'$ be the other node in $x,y$, i.e., $v'\in\{x,y\}\setminus\{v\}$.
We have:
{
\small
\begin{align*}
&|h_v^{(i)}-h_v'^{(i)}|\\
\leq & \max \left( |r_v^{(0)}-r_v'^{(0)}|+\frac{1-\alpha}{2}\cdot \left( |h_v^{(i-1)}-h_v'^{(i-1)}|+\frac{h_{v'}'^{(i-1)}}{d'(v')} +\sum_{u:(u,v)\in E}\left| \frac{h_u^{(i-1)}}{d(u)}-\frac{h_u'^{(i-1)}}{d'(u)}\right|\right),|(d(v)-d'(v)) T_v|\right)\\
\leq &   |r_v^{(0)}-r_v'^{(0)}|+\frac{1-\alpha}{2}\cdot \left( |h_v^{(i-1)}-h_v'^{(i-1)}|+\sum_{u:(u,v)\in E}\left| \frac{h_u^{(i-1)}}{d(u)}-\frac{h_u'^{(i-1)}}{d'(u)}\right|\right)+\max\left(\frac{1-\alpha}{2}\cdot\frac{h_{v'}'^{(i-1)}}{d'(v')},|(d(v)-d'(v)) T_v|\right) \\
=&\frac{1-\alpha}{2}\cdot  \left(\left|h_v^{(i-1)}-h_v'^{(i-1)}\right|+\sum_{u:(u,v)\in E}\left| \frac{h_u^{(i-1)}}{d(u)}-\frac{h_u'^{(i-1)}}{d'(u)}\right|\right)+ \max\left(\frac{1-\alpha}{2}\cdot\frac{h_{v'}'^{(i-1)}}{d'(v')},T_v\right) \\
 =  &  \frac{1-\alpha}{2}\cdot  \left(\left|h_v^{(i-1)}-h_v'^{(i-1)}\right|+\sum_{u:(u,v)\in E}\left| \frac{h_u^{(i-1)}}{d(u)}-\frac{h_u'^{(i-1)}}{d'(u)}\right|\right)+T_v,
\end{align*}
}
where the  third step follows from $|d(v)-d(v')|=1$ and $r_v^{(0)}=r_v'^{(0)}$ and the last step follows from $\frac{h_{v'}'^{(i-1)}}{d'(v')}\leq T_v$ and $(1-\alpha)/2<1$.
Therefore, 
{\small
\begin{align*}
\|h^{(i)}-h'^{(i)}\|_1
\leq & \left( \left\|h^{(i-1)}-h'^{(i-1)}\right\|_1+\sum_{v\in V}\sum_{u:(u,v)\in E}\left| \frac{h_u^{(i-1)}}{d(u)}-\frac{h_u'^{(i-1)}}{d'(u)}\right|
\right)\cdot \frac{1-\alpha}{2}+T_x+T_y\\
=&\left( \left\|h^{(i-1)}-h'^{(i-1)}\right\|_1
+ \sum_{u\in V\setminus\{x,y\}} \left|h_u^{(i-1)}-h_u'^{(i-1)}\right|\right.\\
&\left.+d(x)\cdot \left| \frac{h_x^{(i-1)}}{d(x)}-\frac{h_x'^{(i-1)}}{d(x)+1}\right|
+d(y)\cdot \left| \frac{h_y^{(i-1)}}{d(y)}-\frac{h_y'^{(i-1)}}{d(y)+1}\right|
\right)\cdot \frac{1-\alpha}{2}+T_x+T_y.
\end{align*}
}
Notice that $d(x)\cdot \left| \frac{h_x^{(i-1)}}{d(x)}-\frac{h_x'^{(i-1)}}{d(x)+1}\right|\leq |h_x^{(i-1)}-h_x'^{(i-1)}|+\frac{h_{x}'^{(i-1)}}{d(x)+1}$.
Similarly, $d(y)\cdot \left| \frac{h_y^{(i-1)}}{d(y)}-\frac{h_y'^{(i-1)}}{d(y)+1}\right|\leq |h_y^{(i-1)}-h_y'^{(i-1)}|+\frac{h_{y}'^{(i-1)}}{d(y)+1}$.
Therefore, we have:
{
\small
\begin{align*}
\|h^{(i)}-h'^{(i)}\|_1\leq& (1-\alpha)\cdot \left( \left\|h^{(i-1)}-h'^{(i-1)}\right\|_1+\frac12\cdot\frac{h_{x}'^{(i-1)}}{d'(x)}+\frac12\cdot\frac{h_{y}'^{(i-1)}}{d'(y)}\right)+T_x+T_y.
\end{align*}
}
According to Observation~\ref{obs:properies_capped_pushflow}, we have $\frac{h_{x}'^{(i-1)}}{d'(x)}\leq T_x$ and $\frac{h_{y}'^{(i-1)}}{d'(y)}\leq T_y$.
Therefore, $\|h^{(i)}-h'^{(i)}\|_1\leq (1-\alpha)\cdot\left\|h^{(i-1)}-h'^{(i-1)}\right\|_1+\frac{3-\alpha}2\cdot (T_x+T_y)$.
Since $\|h^{(0)}-h'^{(0)}\|_1=0$, we have $\|h^{(R)}-h'^{(R)}\|_1\leq \frac{3-\alpha}2\cdot (T_x+T_y)\cdot(1-(1-\alpha)^R)/\alpha$.
Hence, $\|p^{(R)}-p'^{(R)}\|_1\leq \frac{3-\alpha}2\cdot (T_x+T_y)\cdot(1-(1-\alpha)^R)$.
If $G'\in\Gamma_s(G)$, i.e., $s\not=x,y$, $T_x+T_y=2\cdot \frac{\sigma}{ (3-\alpha)\cdot (1-(1-\alpha)^R)}$.
If non-joint sensitivity is considered and $G'\not\in\Gamma_s(G)$, i.e., $s\in\{x,y\}$, $T_x+T_y\leq 2\cdot \frac{\sigma}{ (3-\alpha)\cdot (1-(1-\alpha)^R)}$.
We conclude the proof.
\end{proof}

The following lemma shows the running time and the approximation guarantee of Algorithm~\ref{alg:capped_push_flow}.
See~\refappendix{sec:detailed_proof_of_correctness_of_cap_ppr} for the proof.
\begin{lemma}\label{lem:correctness_cap_ppr}
\textsc{PushFlowCap}$(G,s,\alpha,\xi,\sigma, \tp)$ runs in $O(|E|\log(1/\xi)/\alpha)$ time.
Furthermore, if the minimum degree of $G$ is at least $\max\left(1/(\alpha\cdot T_s),\sqrt{1/(\alpha\cdot T_u)}\right)$ $(u\not=s)$, the output of \textsc{PushFlowCap}$(G,s,\alpha,\xi,\sigma, \tp)$ is exactly the same as the output of \textsc{PushFlow}$(G,s,\alpha,\xi)$, i.e., it is a $\xi$-approximate PPR vector for $\p(s)$.

\end{lemma}

\begin{remark}
According to the choice of $T_s$ and $T_u$, the above lemma shows that if $\sigma \geq \Omega_{\alpha}(1/D^2)$ (resp. $\sigma \geq \Omega_{\alpha}(1/D)$), the output of \textsc{PushFlowCap}$(G,s,\alpha,\xi,\sigma, \tp=\joint)$ 
(resp. for $\tp = \njoint$) is a $\xi$-approximate PPR vector for $\p(s)$. 
This is near optimal since we show in~\refappendix{sec:tight_bound_of_sensitivity_of_ppr} that the joint sensitivity (resp. non-joint sensitivity) of the ground truth PPR of a graph with minimum degree $D$ can be at least $\Omega(1/D^2)$ (resp. $\Omega(1/D)$). 
\end{remark}

We apply Laplace mechanism to Algorithm~\ref{alg:capped_push_flow} and obtain a DP or joint DP PPR algorithm for general input graphs (see Algorithm~\ref{alg:dp_capped_push_flow}).
We conclude the theoretical guarantees of our algorithm.

\begin{algorithm}[h]
	\footnotesize
	\begin{algorithmic}[1]\caption{\textsc{DPPushFlowCap}$(G,s,\alpha,\xi,\sigma,\varepsilon,\tp)$}\label{alg:dp_capped_push_flow}
	\STATE \textbf{Input:} Graph $G=(V,E)$, source $s\in V$, teleport probability $\alpha$, precision $\xi$, sensitivity parameter $\sigma$, DP parameter $\varepsilon$, and $\tp\in \{\joint,\njoint\}$ indicating whether joint $\varepsilon$-DP or $\varepsilon$-DP is required.
	\STATE \textbf{Output:} $\varepsilon$-DP approximate PPR vector for $\p(s)$.
    \STATE $p\gets \textsc{PushFlowCap}(G,s,\alpha,\xi,\sigma, \tp)$.
    \STATE Let $Y_1,Y_2,\cdots,Y_n$ drawn independently from $\lap\left(\frac{\sigma}{\varepsilon}\right)$
    \STATE Output $p+(Y_1,Y_2,\cdots,Y_n)$.
	\end{algorithmic}
\end{algorithm}

\begin{corollary}[Joint DP and DP PPR]\label{cor:DPPushFlowCap}
The family of (personalized) algorithms $\{\mathcal{A}_s(G):=\textsc{DPPushFLowCap}(G,s,\alpha,\xi,\sigma,\varepsilon,\joint)\mid s\in V\}$ is joint $\varepsilon$-DP, and \textsc{DPPushFLowCap}$(G,s,\alpha,\xi,\sigma,\varepsilon,\njoint)$ is $\varepsilon$-DP with respect to $G$ for any $s\in V$.
In addition, if the input graph $G$ has a minimum degree at least $D$, the joint $\varepsilon$-DP (resp. $\varepsilon$-DP) output is a $\left(\xi, O_{\alpha,\varepsilon}(\sigma\ln \frac{n}{\delta})\right)$-approximate PPR with probability at least $1-\delta$ for any $\delta\in(0,1)$ when $\sigma \geq \Omega_{\alpha}(1/D^2)$ (resp. $\sigma\geq \Omega_{\alpha}(1/D)$).
\end{corollary}
According to above corollary, we want the sensitivity parameter $\sigma$ to be as small as possible since smaller $\sigma$ leads to smaller additive error.
On the other hand, too small $\sigma$ may make the approximate PPR obtained by \textsc{PushFlowCap} (Algortihm~\ref{alg:capped_push_flow}) deviate  from the ground truth PPR.
Therefore, if a graph has a minimum degree $D$, our joint $\varepsilon$-DP PPR (resp. $\varepsilon$-DP PPR) algorithm provides the best theoretical approximation guarantees when $\sigma=\Theta_{\alpha}(1/D^2)$ (resp. $\sigma =\Theta_{\alpha}(1/D)$). 
 This matches the lower bound shown in~\refappendix{sec:tight_bound_of_sensitivity_of_ppr}:  the ground truth PPR has sensitivity $\Omega_{\alpha}(1/D)$ and joint sensitivity $\Omega_{\alpha}(1/D^2)$. 
 Thus, our results are actually theoretically optimal up to constant factors. 
 The implication of Corollary 4.6 is hence tight, and it is impossible to have a good theoretical approximation guarantee if $\sigma < o_{\alpha}(1/D)$ for $\varepsilon$-DP or $\sigma < o_{\alpha}(1/D^2)$ for joint $\varepsilon$-DP.
In the experimental section we show, however, that in practice the algorithm performs well for a vast range of the parameter setting.

\vspace{-0.1in}
\section{Differentially Private Graph Embeddings}\label{sec:DP_instant_embedding}
\vspace{-0.05in}
As we discussed before, PPR has plenty of applications in graph learning~\cite{klicpera2018predict,postuavaru2020instantembedding,bojchevski2020scaling}. In this section, to show the potentiality of our algorithms, we focus on a recent example of the use of PPR for computing graph embedding. 
We consider Instant\-Embedding~\cite{postuavaru2020instantembedding} (see Algorithm~\ref{alg:instant_embedding}) which is one practical PPR-based graph embedding algorithm. The algorithm proceeds computing the PPR vector of $s$ and then hashing them to obtain an embedding in $\mathbb{R}^k$ for the node $s$. Using our {\it DP} PPR algorithm output as input to Instant\-Embedding leads trivially to a DP embedding algorithm. However, in this section, we show a better implementation which reduces the amount of noise added using in a slightly more sophisticated way our {\it sensitivity bounded} PPR algorithm. We think that this technique could be adapted to other uses of PPR.

First we provide a sensitivity bound for the InstantEmbedding algorithm when applying the sensitivity-bounded PPR. As a result, we show how to obtain differentially private Instant\-Embedding. 
The proof of the following theorem (the sensitivity bound) can be found in~\refappendix{sec:sensitivity_of_embedding}.
\begin{algorithm}[h]
	\footnotesize
	\begin{algorithmic}[1]\caption{\textsc{InstantEmbedding}$(p,k)$}\label{alg:instant_embedding}
	\STATE \textbf{Input:} An approximate PPR vector $p$ for $\p(s)$, dimension $k$, and uniform random hash functions $h_k:V\rightarrow[k],h_{sgn}:V\rightarrow\{-1,1\}$.
	\STATE \textbf{Output:} Embedding vector $w\in\mathbb{R}^k$.
	\STATE Initialize $w\gets \mathbf{0}^k$.
	\FOR{$v\in V$}
	\STATE $w_{h_k(v)}\gets  w_{h_k(v)} + h_{sgn}(v)\cdot \max(\log(p_v\cdot n),0)$.
	\ENDFOR
	\STATE Output $w$.
	\end{algorithmic}
\end{algorithm}
\vspace{-0.1in}
\begin{theorem}[Sensitivity-bounded InstantEmbedding]\label{thm:sensitivity_instantembedding}
Consider two neighboring graphs $G=(V,E),G'=(V,E')$ and source node $s$.
Let $p,p'$ be approximate PPR vectors for $\p(s)$ with respect to $G$ and $G'$ respectively.
Let $w,$ $w'$ be the output of \textsc{InstantEmbedding}$(p,k)$ and \textsc{InstantEmbedding}$(p',k)$ respectively.
Let $m$ be the number of non-zero entries of $p-p'$.
Then, $\|w-w'\|_1\leq m\cdot \log \left(1+\|p-p'\|_1\cdot \frac{n}{m}\right)$.
\end{theorem}
\vspace{-0.05in}

Non-zero entries of $p-p'$ in the above theorem can be at most $n$.
Thus, $\|w-w'\|_1$ is always at most $\|p-p'\|_1\cdot n$.
If we compute $p\gets \textsc{PushFlowCap}(G,s,\alpha,\xi,\sigma)$, according to Theorem~\ref{thm:sensitivity_capped_pushflow}, the (joint) sensitivity of $p$ is at most $\sigma$.
Then, according to Theorem~\ref{thm:sensitivity_instantembedding}, the (joint) sensitivity of output $w$ of Algorithm~\ref{alg:instant_embedding} (using as approximate ppr $p$) is at most $\sigma\cdot n$. 
This allows us to obtain a (joint) $\varepsilon$-DP version of InstantEmbedding by using the Laplace mechanism (Theorem~\ref{thm:laplace_mechanism}) to add $\lap\left(\log(1+\sigma)\cdot n/\varepsilon\right)$ to each entry of $w$.

%% file: exp.tex
\vspace{-0.1in}
\section{Experiments}\label{sec:exp}
\vspace{-0.1in}
In this section, we complement our theoretical analysis with an experimental study of our algorithms in terms of the accuracy of the PPR ranking and node classification. To foster the replicability of our work, we released open source a version of our code.\footnote{ \url{https://github.com/google-research/google-research/tree/master/private_personalized_pagerank} contains our codebase.}  In this experimental section we only consider the joint DP setting due to its practicality in  personalization applications typical of PPR and its better performance. Hence all private algorithms reported are joint DP.

\vspace{-0.1in}
\paragraph{Hyperparameter settings.}
Unless otherwise specified, we ran all experiments that use our \textsc{PushFlowCap} algorithm (or the non-private \textsc{PushFlow}  algorithm~\cite{andersen2007using}) to obtain PPR rankings consistent to the setting commonly used in the literature  of $\alpha=0.15.$\footnote{The algorithms use the {\it lazy} random walk~\cite{andersen2007using}, so we set $\alpha = 0.08$ to match the non-lazy $\alpha$ of $0.15$.} Moreover, we set $\xi$ so that the  number of iterations $R = 100$ in all algorithms. \ We use embedding dimensionality $k=256$.
We observe that our algorithm's utility is not strongly affected by the sensitivity parameter $\sigma$. For simplicity, we always set $\sigma=10^{-6}$, as this parameter generalized across different datasets tested. 
\vspace{-0.1in}
\paragraph{Additional Heuristic.}
We made a simple modification to our algorithm that we empirically observed to increase the performances without affecting the privacy guarantees. For the joint DP case, we push all flow from the source directly (bypassing the algorithm description). This corresponds to making the initial residual $r^{(0)}_s$ of source $s$ to $0$, and that of the neighbors of $s$ to $(1 - \alpha)^2 / d(s)$; and initializing the PPR $p^{(0)}$ of the source (respectively, neighbors) to $\alpha$ (resp., $\alpha\cdot(1-\alpha) / d(s)$). We report the experiments using this improvement.

\vspace{-0.05in}
\paragraph{Baselines.}
To the best of our knowledge, this is the first PPR paper with differential privacy. 
To compare with relevant joint DP baselines we use the standard randomized response~\cite{dwork2014algorithmic} (edge-flipping) baseline applied to the graph. Given the source node $s$, for each (unordered) pair $\{u,v\}$ of nodes ($u,v\neq s$) we apply the randomized response mechanism on the corresponding entry in the adjacency matrix: with probability $p = 2 / (1 + \exp( \epsilon))$, the entry is replaced by a uniform at random $\{0,1\}$, otherwise we keep the true entry. This results in a joint $\epsilon$-DP output over which we run the non-private PPR algorithm~\cite{andersen2007using}. 
Note that this baseline, for $\epsilon=\Theta(1)$, requires $\Theta(n^2)$ time to generate a DP adjacency matrix, and makes the output matrix dense. This limits the applicability of the baseline mechanism, whereas our approach remains scalable for larger graphs. 
\vspace{-0.05in}

\vspace{-0.1ex}
\begin{table}%[8]{r}{50ex}
\footnotesize
\centering
\vspace{-4mm} 
\newcolumntype{C}{>{\centering\arraybackslash}X}
\caption{Dataset characteristics: number of vertices \(|V|\), edges \(|E|\), node labels \(|\mathcal{L}|\); avg.\ degree; density as \(|E|/\binom{|V|}{2}\).\label{tbl:datasets}}
\begin{tabularx}{\linewidth}{p{1.5cm}CCC>{\centering\arraybackslash}p{1cm}>{\centering\arraybackslash}p{1.4cm}}
\toprule
\multicolumn{1}{c}{} & \multicolumn{3}{c}{\textbf{Size}} & \multicolumn{2}{c}{\textbf{Statistics}} \\
\cmidrule(lr){2-4}\cmidrule(lr){5-6}
\emph{dataset} & \(|V|\) & \multicolumn{1}{c}{\(|E|\)} & \(|\mathcal{L}|\) & \mbox{Avg.\ deg.} & Density \\
    \midrule
{\small POS} & 5k & 185k & 40 & 38.7 &  \mbox{\(8.1 \times 10^{-3}\)} \\
{\small Blogcatalog} & 10k & 334k & 39 & 64.8 &  \mbox{\(6.3 \times 10^{-3}\)} \\
\bottomrule
\end{tabularx}
\end{table}
\vspace{-0.1in}
\paragraph{Datasets.}
We experiment on 2 publicly available datasets
available from~\cite{grover2016node2vec,tsitsulin2018verse}.
Table~\ref{tbl:datasets} reports basic statistics about these datasets.
POS is a word co-occurrence network built from Wikipedia.
BlogCatalog a social network of bloggers from the blogcatalog website.

\vspace{-0.15in}
\subsection{PPR Approximation Accuracy}\label{sec:ppr-approximation}
\vspace{-0.15in}
We verify that our algorithms can rank the nodes of real-world graphs in a private way.
We use as seeds nodes with degree $50$ or more and compute ``ground-truth'' PPR values with the standard power iteration algorithm. Then, we run \textsc{DPPushFlowCap} algorithm in the joint DP setting, rerunning the algorithm $100$ times independently for each seed.
Figure~\ref{fig:ppr-approximation} presents the results on the datasets studied.
We examine the performance of Joint DP PPR from two standard metrics: Recall@k and normalized discounted cumulative gain (NDCG)~\cite{jarvelin2002cumulated}. We select Recall@100 and NDCG since the top PPR values are the most important in practical applications.
We observe that the joint DP rankings offer comparable top-k predictions across datasets and metrics tested than the private baseline.

\input{figures/approximation}
\input{figures/embedding_neurips}

\vspace{-0.15in}
\subsection{Node Classification via Private Embeddings}\label{sec:node-classification-real}
\vspace{-0.15in}
Last, we examine the performance of our joint DP embedding algorithm.
We follow the procedure of~\cite{perozzi2014deepwalk,grover2016node2vec} and evaluate our embeddings on a node classification task in real-world graphs.
We report the results in Figure~\ref{fig:embeddings}. Here, Joint DP represents our DP InstantEmbedding algorithm; DP Baseline + InstantEmbedding represents the obvious baseline of computing Baseline DP PPR followed by the hashing procedure; Non-DP InstantEmbedding is the original non-DP embedding and Random represents a uniform random embedding. Our joint DP algorithm has significantly better performance than the baseline and has non-trivial Micro-F1 even for small $\varepsilon$. In contrast, the baseline does not extract any useful information from the graph in this range of $\varepsilon$.

%% file: figures/approximation.tex
\begin{figure}[t]
\centering
\begin{tikzpicture}
\begin{groupplot}[group style={
                      group name=myplot,
                      group size= 4 by 1,vertical sep=1.5cm, horizontal sep=1.15cm},
                      height=3cm,
                      width=.285\linewidth,
                      xlabel=$\varepsilon$,
                      xmin=0.7,
                      xmax=10,
                      ticklabel style = {font=\small},
                      y label style={at={(axis description cs:-0.25,.5)},anchor=south},
                      title style={at={(1.25,-.5)},anchor=north},
]
\nextgroupplot[
    xmin=0,
    xmax=5,
    ymax=100,
    ylabel=Recall@100,
 	title =\textbf{POS dataset},
]
\addplot[very thick,color=cycle3,mark=square*,mark repeat=10] table[y=dpppr-prepush,x=epsilon]  {data_final/pos-recall.tex};
\addplot[very thick,color=cycle10,mark=*,mark repeat=1] table[y=baseline,x=epsilon] {data_final/pos-recall.tex};
\nextgroupplot[
    ylabel=NDCG,
    ymax=100,
    xmin=0,
    xmax=5,
]
\addplot[very thick,color=cycle3,mark=square*,mark repeat=10] table[y=dpppr-prepush,x=epsilon]  {data_final/pos-ndcg.tex};
\addplot[very thick,color=cycle10,mark=*,mark repeat=1] table[y=baseline,x=epsilon] {data_final/pos-ndcg.tex};
\nextgroupplot[
    xmin=0,
    xmax=5,
    ymax=100,
    ylabel=Recall@100,
 	title =\textbf{Blog dataset},
]
\addplot[very thick,color=cycle3,mark=square*,mark repeat=10] table[y=dpppr-prepush,x=epsilon]  {data_final/blogcatalog-recall.tex};
\addplot[very thick,color=cycle10,mark=*,mark repeat=1] table[y=baseline,x=epsilon] {data_final/blogcatalog-recall.tex};
\nextgroupplot[
    ylabel=NDCG,
    ymax=100,
    xmin=0,
    xmax=5,
	legend style={at={(1.,1.15)},anchor=south east},
    legend entries={Joint DP PPR, Edge flipping joint DP baseline},
    legend columns=2,
    legend cell align=left,
 legend style={nodes={scale=0.9, transform shape}},
]
\addplot[very thick,color=cycle3,mark=square*,mark repeat=10] table[y=dpppr-prepush,x=epsilon]  {data_final/blogcatalog-ndcg.tex};
\addplot[very thick,color=cycle10,mark=*,mark repeat=1] table[y=baseline,x=epsilon] {data_final/blogcatalog-ndcg.tex};

\end{groupplot}
\end{tikzpicture}
\vspace{-0.3in}
\caption{\footnotesize \label{fig:ppr-approximation}PPR approximation on two datasets.}

\end{figure}
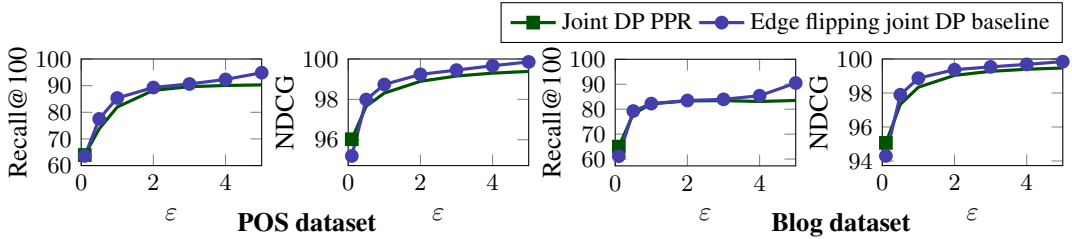

%% file: figures/embedding_neurips.tex
\begin{figure}[t]%[16]{r}{40ex}
\centering
\vspace{-0.2in}
\begin{tikzpicture}
\begin{axis}[
 height=3.75cm,
 width=0.7\linewidth,
 xlabel=$\varepsilon$,
 xmin=0.01,
 xmax=10,
 xmode=log,
 ticklabel style = {font=\small},
 y label style={at={(axis description cs:-0.15,.5)},anchor=south},
 ylabel=Micro F1 $\times 100$,
 legend style={nodes={scale=0.9, transform shape}},
 title style={at={(0.5,-.4)},anchor=north},
 legend style={at={(1,1.1)},anchor=south east},
 legend entries={Joint DP PPR, Joint DP Baseline + InstantEmbedding, Random, Non-private InstantEmbedding},
 legend columns=1,
 legend cell align=left,
]
\addplot[very thick,mark=*,color=cycle4] table {
0.01 16.6939
0.1 23.4103
0.5 23.5266
1 23.5168
2 23.5246
3 23.5296
4 23.5248
5 23.6052
6 23.5974
7 23.6309
8 23.5970
9 23.5167
10 23.5843
};
\addplot[very thick,color=cycle9,mark=square] table {
0.01 15.4598
0.1 15.5202
1 14.0937
2 12.2169
3 11.8794
4 12.7143
5 14.7460
6 18.5687
7 22.5679
8 26.5911
9 29.0269
10 30.3027
};
\addplot[very thick,color=cycle7] table {
0.01 12.4846
10.0 12.4846
};

\addplot[very thick,color=cycle6] table {
0.01 30.6802
10.0 30.6802
};
\end{axis}
\end{tikzpicture}
\vspace{-3.4ex}
\caption{\footnotesize \label{fig:embeddings}Private  embeddings introduced in this paper outperform other competitors and achieve much higher performance on a tight privacy budget. }

\end{figure}
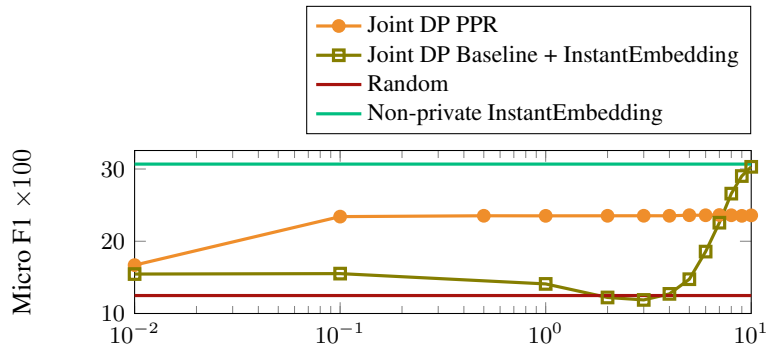

%% file: conclusion.tex
\vspace{-0.15in}
\section{Conclusion}
\vspace{-0.15in}
We showed that it is possible to approximate PPR  with DP. We believe that the techniques developed for bounding the sensitivity of PPR can find applications in other areas of graph-learning.  As a future work, we would like to extend the use of our DP PPR algorithm to more machine learning tasks. 

\vspace{-0.1in}
\paragraph{Societal impact and limitations}
Our work focuses on developing DP algorithms. DP provides strong protection, but has limitations \cite{dwork2014algorithmic}. Moreover, while privacy is a requirement of a responsible computational system, it is not the only one.

\section{Acknowledgements and corrections from the published version}

We thank Simone Antonelli and Aleksandar Bojchevski for identifying an error in a prior version of the paper.

For this reason, we made a series of corrections and updates to the paper after the publication at NeurIPS 2022.

\begin{enumerate}
    \item We have corrected  all experimental results to account for an error in the code.
    \item We strengthen the theoretical analysis of the algorithm resulting in slightly improved constants.
    \item We introduce a heuristic that improves the empirical results in real-world graphs while preserving the same joint-DP privacy guarantees.
    \item We made some other small text adjustments.
\end{enumerate}

%% file: app.tex
\section{Tightness of the Sensitivity Bound for PPR}\label{sec:tight_bound_of_sensitivity_of_ppr}
In this section, we show examples that the sensitivity of a PPR vector of a graph with minimum degree $D$ can be $\Omega(1/D)$ and the sensitivity for joint DP can be $\Omega(1/D^2)$.
For simplicity, consider $\alpha = 0.5$. 
Consider a clique with $D+1$ nodes.
Thus, each node has degree $D$.
Let us choose an arbitrary node as the source node.
Let $p$ be the PPR vector for $s$.
Then it is easy to verify $p_s=\frac{2D+1}{3D+1}$ and $p_v=\frac{1}{3D+1}$ for $v\not=s$.

Suppose we remove the edge between node $s$ and node $x$.
Let $p'$ be the new PPR vector for $s$.
We can verify that $p'_s=\frac{6D+5}{9D+6},p'_x=\frac{1}{9D+6}$ and $p'_v=\frac{D}{3D^2-D+2}$ for $v\not=s,x$.
It is easy to see that $|p_x-p'_x|=|1/(3D+1)-1/(9D+6)|$ is already $\Omega(1/D)$.

Suppose we remove the edge between nodes $x,y\not=s$. 
Let $p''$ be the new PPR vector for $s$.
We can verify that $p''_s= \frac{6D^3+D^2-5D}{9D^3-7D-2},p''_x=p''_y=\frac{1}{3D+2},$ and $p''_v=\frac{3D^2-D}{9D^3-7D-2}$ for $v\not=x,y$.
It is easy to see that $|p_x-p''_x|=|1/(3D+1)-1/(3D+2)|$ is already $\Omega(1/D^2)$.

\section{Proof of Lemma~\ref{lem:correctness_of_standard_ppr}}\label{sec:detailed_proof_of_correctness_of_standard_ppr}
As defined in Section~\ref{sec:preli_ppr}, let $W$ be the lazy random walk matrix of $G$.
\begin{lemma}[Linearity of PPR vector~\cite{andersen2007using}]\label{lem:linearity_of_ppr}
Given any $x\in\mathbb{R}^n$, $\p(x)\cdot W=\p(x\cdot W)$.
Given any $x,y\in\mathbb{R}^n$, $\p(x+y)=\p(x)+\p(y)$.
Given any $x\in\mathbb{R}^n,a\in\mathbb{R}$, $\p(a\cdot x)=a\cdot \p(x)$.
\end{lemma}
Similar to Lemma 3.4 of~\cite{andersen2007using}, we have the following lemma.
\begin{lemma}[Push flow operation]\label{lem:push_flow_guarantee}
Consider $n$-node graph $G=(V,E)$.
Let $s\in\mathbb{R}^n$ be the starting distribution vector.
Given vectors $p\in\mathbb{R}^n$ and $r\in\mathbb{R}^n$ satisfying $p=\p(s-r)$.
For any $x\in\mathbb{R}^n$, if $p',r'$ are computed as the following:
\begin{enumerate}
\item $\forall v\in V,p'_v = p_v+\alpha\cdot x_v$,
\item $\forall v\in V,r'_v = r_v-x_v+(1-\alpha)/2\cdot \left(x_v+\sum_{(v,u)\in E}x_u/d(u)\right)$,
\end{enumerate}
then we have $p'=\p(s-r')$.
\end{lemma}
\begin{proof}
We can rewrite the computation of $p',r'$ as the following:
\begin{enumerate}
    \item $p'=p+\alpha\cdot x$,
    \item $r'=r-x+(1-\alpha)\cdot x\cdot W$.
\end{enumerate}
Thus, we have:
\begin{align*}
\p(r)&=\p(r-x)+\p(x)\\
&=\p(r-x) + \alpha x + (1-\alpha)\p(x)W\\
&=\p(r-x) + \alpha x + \p((1-\alpha)xW)\\
&=\p(r-x+(1-\alpha)xW) +\alpha x\\
&=\p(r')+p'-p,
\end{align*}
where the first equality follows from linearity (Lemma~\ref{lem:linearity_of_ppr}), the second equality follows from the definition of PPR $\p(x)$, the third and the forth equalities follow from linearity (Lemma~\ref{lem:linearity_of_ppr}) again, and the last equality follows from the definition of $p'$ and $r'$.

Therefore, we have $p+\p(r)=p'+\p(r')$.
Due to linearity (Lemma~\ref{lem:linearity_of_ppr}), we have $p-\p(s-r)=p'-\p(s-r')$.
Since $p=\p(s-r)$, $p'=\p(s-r')$.
\end{proof}
We are now able to prove Lemma~\ref{lem:correctness_of_standard_ppr}.
\begin{proof}[Proof of Lemma~\ref{lem:correctness_of_standard_ppr}]
Let us first consider the running time.
The algorithm has $R=O(\ln(1/\xi)/\alpha)$ iterations. 
In each iteration, the algorithm can visit each neighbor of each node at most once.
Thus, the running time needed for each iteration is $O(|E|)$.
The total running time is $O(|E|\cdot \ln(1/\xi)/\alpha)$.

Next, let us focus on the accuracy of the output.
Notice that $\forall i\in[R]$, we have
\begin{enumerate}
\item $\forall v\in V,p^{(i)}_v=p^{(i-1)}_v+\alpha\cdot r_v^{(i-1)}$,
\item $\forall v\in V,r^{(i)}_v=(1-\alpha)/2\cdot \left(r_v^{(i-1)}+\sum_{(v,u)\in E}r_u^{(i-1)}/d(u)\right)$.
\end{enumerate}
Since $p^{(0)}=\p(s-r^{(0)})$, according to Lemma~\ref{lem:push_flow_guarantee}, $\forall i\in [R],p^{(i)}=\p(s-r^{(i)})$.
Next we want to show that each entry of $r^{(R)}$ is non-negative and is at most $\xi$.
It is easy to verify that any operation during the algorithm will not create a negative value of any entry of $r^{(j)}$ for $j\in\{0,1,\cdots,R\}$.
Next consider the maximum value of $r^{(j)}$ for $j\in\{0,1,\cdots,R\}$.
The proof is by induction. 
It is obvious that $\|r^{(0)}\|_1=1$.
Consider $j>0$. We have $\|r^{(j)}\|_1=\sum_{v\in V}r^{(j)}_v=(1-\alpha)/2\cdot \sum_{v\in V}r^{(j-1)}_v+(1-\alpha)/2\cdot \sum_{v\in V}\sum_{u:(u,v)\in E} r^{(j-1)}_{u}/d(u)=(1-\alpha)/2\cdot \sum_{v\in V}r^{(j-1)}_v+(1-\alpha)/2\cdot \sum_{u\in V} d(u)\cdot r^{(j-1)}_{u}/d(u)$ = $(1-\alpha)\cdot \|r^{(j-1)}\|_1\leq (1-\alpha)^j$.
Thus, $\|r^{(R)}\|_{\infty}\leq \|r^{(R)}\|_1\leq (1-\alpha)^R\leq \xi$.

According to Definition~\ref{def:xi_approx_ppr}, since $p^{(R)}=\p(s-r^{(R)})$ and each entry of $r^{(R)}$ is non-negative and is at most $\xi$, $p^{(R)}$ is an $\xi$-approximate PPR vector for $\p(s)$.
\end{proof}

\section{Bound of $r'_x,r'_y$ in the Proof of Theorem~\ref{thm:sensitivity_standard_pushflow}}\label{sec:proof_of_claims}
\begin{claim}\label{cla:total_residual}
$\forall j\in[R],\|r'^{(j)}\|_1\leq (1-\alpha)^j$.
\end{claim}

\begin{proof}[Proof of Claim~\ref{cla:total_residual}]
The proof is by induction. 
It is obvious that $\|r'^{(0)}\|_1=1$.
Consider $j>0$. We have $\|r'^{(j)}\|_1=\sum_{v\in V}r'^{(j)}_v=(1-\alpha)/2\cdot \sum_{v\in V}r'^{(j-1)}_v+(1-\alpha)/2\cdot \sum_{v\in V}\sum_{u:(u,v)\in E'} r'^{(j-1)}_{u}/d'(u)=(1-\alpha)/2\cdot \sum_{v\in V}r'^{(j-1)}_v+(1-\alpha)/2\cdot \sum_{u\in V} d'(u)\cdot r'^{(j-1)}_{u}/d'(u)$ = $(1-\alpha)\cdot \|r'^{(j-1)}\|_1\leq (1-\alpha)^j$.
\end{proof}

\begin{claim}\label{cla:each_residual}
$\forall j\in [R],\forall u\not=s$, $r'^{(j)}_u\leq (1-\alpha)^j/D$.
\end{claim}
\begin{proof}[Proof of Claim~\ref{cla:each_residual}]
The proof is by induction.
When $j=0$, $\forall u\not=s,r'^{(0)}_u = 0$.
Consider $j>0$.
We have $r'^{(j)}_u\leq (1-\alpha)/2\cdot r'^{(j-1)}_u+(1-\alpha)/2\cdot \sum_{v\in V}r'^{(j-1)}_v/d'(v)\leq (1-\alpha)/2\cdot r_u'^{(j-1)}+(1-\alpha)/2\cdot \|r'^{(j-1)}\|_1/D$.
According to Claim~\ref{cla:total_residual} and induction hypothesis, we have $r'^{(j)}_u\leq (1-\alpha)\cdot (1-\alpha)^{j-1}/D=(1-\alpha)^j/D$.
\end{proof}

\section{Na{\"i}ve DP PPR for High Degree Graphs}\label{sec:DPPushFlow}

\begin{algorithm}[h]
	\small
	\begin{algorithmic}[1]\caption{\textsc{DPPushFlow}$(G,s,\alpha,\xi,\varepsilon)$}\label{alg:dp_standard_push_flow}
	\STATE \textbf{Input:} Graph $G=(V,E)$ with minimum degree at least $D$, source node $s\in V$, teleport probability $\alpha$, precision $\xi$, DP parameter $\varepsilon$.
	\STATE \textbf{Output:} $\varepsilon$-DP approximate PPR vector for $\p(s)$.
    \STATE $p\gets \textsc{PushFlow}(G,s,\alpha,\xi)$.
    \STATE // If considering joint $\varepsilon$-DP:
    \STATE Let $Y_1,Y_2,\cdots,Y_n$ drawn independently from $\lap\left(\frac{2(1-\alpha)}{\varepsilon\cdot \alpha\cdot D^2}\right)$
    \STATE // Otherwise for $\varepsilon$-DP:
    \STATE Let $Y_1,Y_2,\cdots,Y_n$ drawn independently from $\lap\left(\frac{2(1-\alpha)}{\varepsilon\cdot \alpha\cdot D}\right)$
    \STATE Output $p+(Y_1,Y_2,\cdots,Y_n)$.
	\end{algorithmic}
\end{algorithm}

\begin{corollary}\label{cor:DPPushFlow}
Suppose the input graph $G$ is guaranteed to have minimum degree at least $D$. 
For any given source node $s$, the output of \textsc{DPPushFLow}$(G,s,\alpha,\xi,\varepsilon)$ is $\varepsilon$-DP and is a $\left(\xi, O_{\alpha,\varepsilon}(D^{-1}\ln \frac{n}{\delta})\right)$-approximate PPR for $\p(s)$ with probability at least $1-\delta$ for any $\delta\in(0,1)$.
The family of (personalized) algorithms $\{\mathcal{A}_s(G):=\textsc{DPPushFLow}(G,s,\alpha,\xi,\varepsilon)\mid s\in V\}$ is joint $\varepsilon$-DP and the output of $\mathcal{A}_s(G)$ is a $\left(\xi, O_{\alpha,\varepsilon}(D^{-2}\ln \frac{n}{\delta})\right)$-approximate PPR for $\p(s)$ with probability at least $1-\delta$ for any $\delta\in(0,1)$.
\end{corollary}
\begin{proof}[Proof of Corollary~\ref{cor:DPPushFlow}]
The $\varepsilon$-DP guarantee follows from Laplace mechanism (Theorem~\ref{thm:laplace_mechanism}).
The sensitivity bound is given by Theorem~\ref{thm:sensitivity_standard_pushflow}.
Next, consider the accuracy of Algorithm~\ref{alg:dp_standard_push_flow}.
According to Lemma~\ref{lem:correctness_of_standard_ppr}, the vector $p$ is a $\xi$-approximate PPR vector. 
Let $\Delta$ be the sensitivity of $p$, i.e., $\Delta=\frac{2\cdot(1-\alpha)}{\alpha\cdot D^2}$ if joint DP is considered, and $\Delta=\frac{2\cdot(1-\alpha)}{\alpha\cdot D}$ otherwise.
Consider $i\in[n]$. 
By Fact~\ref{fac:laplace_bound}, with probability at least $1-\delta/n$, $|Y_i|\leq \frac{\Delta}{\varepsilon}\cdot \ln(n/\delta)$.
By taking union bound over $i\in[n]$, with probability at least $1-\delta$,  $\max_{i\in[n]} |Y_i|\leq \frac{\Delta}{\varepsilon}\cdot \ln(n/\delta)$.
Thus, the output of Algorithm~\ref{alg:dp_standard_push_flow} is a $\left(\xi,\frac{\Delta}{\varepsilon}\cdot\ln(n/\delta)\right)$-approximate PPR for $\p(s)$ with probability at least $1-\delta$.
\end{proof}

%\section{Proof of Theorem~\ref{thm:sensitivity_capped_pushflow}}\label{sec:proof_of_sensitivity_of_capped_push_flow}
%
\section{Proof of Lemma~\ref{lem:property_of_h}}\label{sec:bound_of_cumulative_residual}

\begin{proof}[Proof of Lemma~\ref{lem:property_of_h}]
According to the description of Algorithm~\ref{alg:capped_push_flow}, we have $\forall i\in[R],v\in V,r_v^{(i)}=\frac{1-\alpha}{2}\cdot\left(f_v^{(i)}+\underset{u:(u,v)\in E}{\sum}\frac{f_u^{(i)}}{d(u)}\right)$.
Due to $h_v^{(i)}=\sum_{j=1}^i f_v^{(j)}$ and Observation~\ref{obs:properies_capped_pushflow}, $\sum_{j=0}^{i-1} r_v^{(j)} = r_v^{(0)}+\frac{1-\alpha}{2}\cdot\left(h_v^{(i-1)}+\underset{u:(u,v)\in E}{\sum} \frac{h_u^{(i-1)}}{d(u)}\right)$.
For Observation~\ref{obs:properies_capped_pushflow} again, we can conclude $h_v^{(i)}=\min\left(r_v^{(0)}+\frac{1-\alpha}{2}\cdot \left(h_v^{(i-1)}+ \underset{u:(u,v)\in E}{\sum} \frac{h_u^{(i-1)}}{d(u)}\right),d(v)\cdot T_v\right)$.
\end{proof}

\section{Proof of Lemma~\ref{lem:correctness_cap_ppr}}\label{sec:detailed_proof_of_correctness_of_cap_ppr}
\begin{proof}[Proof of Lemma~\ref{lem:correctness_cap_ppr}]
Let us consider the running time. 
The algorithm has $R$ iterations.
In each iteration, the algorithm can visit each neighbor of each node at most once.
Thus, the running time for each iteration is $O(|E|)$.
The total running time is $O(|E|\cdot R)=O(|E|\cdot \ln(1/\xi)/\alpha)$.

Let $D=\max\left(1/(\alpha T_s), \sqrt{1/(\alpha T_u)}\right)$ ($u\not=s$).
In the remaining of the proof, we consider the input graph $G$ with minimum degree at least $D$.
Firstly, we observe that the value $R$ is the same in Algorithm~\ref{alg:standard_push_flow} and Algorithm~\ref{alg:capped_push_flow}, i.e., both algorithms have the same number of iterations.
Then notice that if $f_v^{(i)}$ is equal to $r_v^{(i-1)}$ for all $i$ and $v$, then we have 
$\forall i\in[R]$,
\begin{enumerate}
\item $\forall v\in V,p^{(i)}_v=p^{(i-1)}_v+\alpha\cdot r_v^{(i-1)}$,
\item $\forall v\in V,r^{(i)}_v=(1-\alpha)/2\cdot \left(r_v^{(i-1)}+\sum_{(v,u)\in E}r_u^{(i-1)}/d(u)\right)$.
\end{enumerate}
Therefore, when $f_v^{(i)}$ is equal to $r_v^{(i-1)}$ for all $i$ and $v$, vectors $p^{(i)},r^{(i)}$ in both Algorithm~\ref{alg:standard_push_flow} and Algorithm~\ref{alg:capped_push_flow} are equal respectively.
Thus, when $f_v^{(i)}$ is equal to $r_v^{(i-1)}$ for all $i$ and $v$, the output of Algorithm~\ref{alg:capped_push_flow} is the same as the output of Algorithm~\ref{alg:standard_push_flow}.
It suffices to show that $\forall i\in [R],v\in V,f_v^{(i)}=r_v^{(i-1)}$ when $G$ has minimum degree at least $D$.

The proof is by contradiction.
Consider the first time that $f_v^{(i)}\not=r_v^{(i-1)}$ during the execution of Algorithm~\ref{alg:capped_push_flow}.
We have $d(v)\cdot T_v-h_v^{(i-1)}<r_v^{(i-1)}$.
Since it is the first time $f_v^{(i)}\not=r_v^{(i-1)}$, we know $h_v^{(i-1)}=\sum_{j=0}^{i-2}r_v^{(j)}$ according to Observation~\ref{obs:properies_capped_pushflow}.
Therefore, we know that $\sum_{j=0}^{i-1} r_v^{(j)} > d(v)\cdot T_v$.
If $v\not=s$, according to the choice of $T_v$, we have $\sum_{j=0}^{i-1} r_v^{(j)}>D/(\alpha\cdot D^2)=1/(\alpha D)$  which contradicts to Claim~\ref{cla:each_residual}.
If $v=s$, according to the choice of $T_v$, we have $\sum_{j=0}^{i-1} r_v^{(j)}>D/(\alpha\cdot D)=1/\alpha$  which contradicts to Claim~\ref{cla:total_residual}.
Thus, we always have $\forall i\in [R],v\in V,f_v^{(i)}=r_v^{(i-1)}$ when $G$ has minimum degree at least $D$.
We conclude that the output of Algorithm~\ref{alg:capped_push_flow} is the same as the output of Algorithm~\ref{alg:standard_push_flow} in this case.
According to Lemma~\ref{lem:correctness_of_standard_ppr}, the output is a $\xi$-approximate PPR vector.
\end{proof}

%\section{Proof of Lemma~\ref{lem:property_of_h}}\label{sec:bound_of_cumulative_residual}

\section{Proof of Corollary~\ref{cor:DPPushFlowCap}}
\begin{proof}[Proof of Corollary~\ref{cor:DPPushFlowCap}]
Since the joint sensitivity and non-joint sensitivity of Algorithm~\ref{alg:capped_push_flow} are given by Theorem~\ref{thm:sensitivity_capped_pushflow}.
Thus, the $\varepsilon$-DP and joint $\varepsilon$-DP guarantees follow from the Laplace mechanism (Theorem~\ref{thm:laplace_mechanism}).

Next, consider the accuracy of Algorithm~\ref{alg:dp_capped_push_flow}.
If considering joint $\varepsilon$-DP, according to Lemma~\ref{lem:correctness_cap_ppr}, the output $p$ of \textsc{PushFlowCap} (Algorithm~\ref{alg:capped_push_flow}) is a $\xi$-approximate PPR vector when the minimum degree of $G$ is at least $\sqrt{(2\cdot(2-\alpha))/(\alpha\cdot \sigma)}$.
Otherwise, the output $p$ of \textsc{PushFlowCap} (Algorithm~\ref{alg:capped_push_flow}) is a $\xi$-approximate PPR vector when the minimum degree of $G$ is at least $(2\cdot(2-\alpha))/(\alpha\cdot \sigma)$.
Consider $i\in[n]$. 
By Fact~\ref{fac:laplace_bound}, with probability at least $1-\delta/n$, $|Y_i|\leq \frac{\sigma}{\varepsilon}\cdot \ln(n/\delta)$.
By taking union bound over $i\in[n]$, with probability at least $1-\delta$,  $\max_{i\in[n]} |Y_i|\leq \frac{\sigma}{\varepsilon}\cdot \ln(n/\delta)$.
Thus, the output of Algorithm~\ref{alg:dp_capped_push_flow} is a $\left(\xi,\frac{\sigma}{\varepsilon}\cdot\ln(n/\delta)\right)$-approximate PPR for $\p(s)$ with probability at least $1-\delta$.
\end{proof}

%\section{Proof of Theorem~\ref{thm:dp_sweep_cut}}\label{sec:proof_of_dp_sweep_cut}
%\textcolor{red}{TODO}

\section{Proof of Theorem~\ref{thm:sensitivity_instantembedding}}\label{sec:sensitivity_of_embedding}
\begin{proof}[Proof of Theorem~\ref{thm:sensitivity_instantembedding}]
%Let $p$ and $p'$ be the output of \textsc{PushFlowCap}$(G,s,\alpha,\xi,D)$ and $\textsc{PushFlowCap}(G',s,\alpha,\xi,D)$ respectively.
For $v\in V$, consider $|\max(\log(p_v\cdot n),0)-\max(\log(p'_v\cdot n),0)|$.
There are several cases.
Without loss of generality, we suppose $p_v>p'_v$.
In the first case, if both $p_v,p'_v\leq 1/n$, then we have $|\max(\log(p_v\cdot n),0)-\max(\log(p'_v\cdot n),0)|=0\leq \log(1+|p_v-p'_v|\cdot n)$.
In the second case, $p_v> 1/n$ and $p'_v\leq 1/n$. 
Then, $|\max(\log(p_v\cdot n),0)-\max(\log(p'_v\cdot n), 0)|=\log(1+(p_v-1/n)\cdot n)\leq \log(1+|p_v-p'_v|\cdot n)$.
In the third case, both $p_v,p'_v>1/n$.
Then $|\max(\log(p_v\cdot n),0)-\max(\log(p'_v\cdot n), 0)|=\log(p_v/p'_v)=\log(1+(p_v-p'_v)/p'_v)\leq \log(1+|p_v-p'_v|\cdot n)$.
By combining the above cases, we always have:
\begin{align}
&|\max(\log(p_v\cdot n),0)-\max(\log(p'_v\cdot n),0)|\notag\\
\leq& \log(1+|p_v-p'_v|\cdot n)\label{eq:bound_of_one_embedding_entry}%\\
%\leq& |p_v-p'_v|\cdot n. 
\end{align}
Now, we are able to bound $\|w-w'\|_1$.
\begin{align*}
&\|w-w'\|_1\\
=&\sum_{i=1}^k |\sum_{v\in V:h_k(v)=i}h_{sgn}(v)\cdot (\max(\log(p_v\cdot n),0)\\
&-\max(\log(p'_v\cdot n),0))|\\
\leq & \sum_{v\in V:p_v\not=p'_v} |\max(\log(p_v\cdot n),0)-\max(\log(p'_v\cdot n),0)|\\
\leq &\sum_{v\in V:p_v\not=p'_v} \log(1+|p_v-p'_v|\cdot n)\\
\leq & m\cdot \log(1+\|p-p'\|_1\cdot n/m),
\end{align*}
where the first inequality follows from triangle inequality, the second inequality follows from Equation~\eqref{eq:bound_of_one_embedding_entry}, and the last inequality follows from the concavity of $\log(\cdot)$.
%According to Theorem~\ref{thm:sensitivity_capped_pushflow}, if $s\not=x,y$, $\|p-p'\|_1\leq \frac{2\cdot (2-\alpha)}{\alpha\cdot D^2}$, which implies $\|w-w'\|_1\leq m\cdot \log \left(1+\frac{2\cdot (2-\alpha)}{\alpha\cdot D^2}\cdot \frac{n}{m}\right)$.
%Otherwise, $\|p-p'\|_1\leq \frac{2\cdot (2-\alpha)}{\alpha\cdot D}$, and thus $\|w-w'\|_1\leq m\cdot \log \left(1+\frac{2\cdot (2-\alpha)}{\alpha\cdot D}\cdot \frac{n}{m}\right)$.
\end{proof}

\section{Differentially Private InstantEmbedding via Sparse Personalized PageRank}
\paragraph{Theoretically Improved InstantEmbedding via Sparse Personalized PageRank.}
%A natural question is whether we can get lower sensivitity of the embedding vector and add smaller noise for DP.
Notice that due to Theorem~\ref{thm:sensitivity_instantembedding}, a sparser approximate PPR may give lower sensitivity of the InstantEmbedding. 
%\todo{this is a bit abrupt. Can we motivate why we need this?}
%In some applications, only large entires of the PPR vector needs to be considered.
%A natural way is to only keep the large entries of the approximate PPR vector.
%But the challenging part is to make the sparsification process differentially private.
Therefore, we show how the embedding algorithm can be further improved by sparsifying the PPR vector in a differentially private way. The sparsification procedure is reported in Algorithm~\ref{alg:sparsification} which keeps large entries with good probabilites and will drop small entries of the input vector.
%In our experiments (see Section~\ref{sec:exp}), we will show that the InstantEmbedding computed via sparse PPR vector has better performance.

\begin{lemma}\label{lem:eps_dp_sparsification}
\textsc{DPSparsification}$(p,\sigma,\varepsilon,\gamma)$ is $\varepsilon$-DP.
\end{lemma}
\begin{proof}
Consider a neighboring vector $p'$ of $p$ i.e., $\|p-p'\|_1\leq \sigma$.
Let $S$ and $S'$ be the output of \textsc{DPSparsification}$(p,\sigma,\varepsilon,\gamma)$ and \textsc{DPSparsification}$(p',\sigma,\varepsilon,\gamma)$ respectively.
For $i\in[n],$  $\Pr[i\in S]/\Pr[i\in S']=\int_{-\infty}^{p_i}\frac{\varepsilon}{2\sigma}e^{-\frac{\varepsilon}{\sigma}|x-\gamma|}\mathrm{d}x/\int_{-\infty}^{p_i'}\frac{\varepsilon}{2\sigma}e^{-\frac{\varepsilon}{\sigma}|x-\gamma|} \mathrm{d}x$.
Notice that $\exp\left(-\frac{\varepsilon}{\sigma}|x-\gamma|\right)/\exp\left(-\frac{\varepsilon}{\sigma}|x+p_i'-p_i-\gamma|\right)\leq \exp\left(\frac{\varepsilon}{\sigma}|p_i'-p_i|\right)$.
Therefore, $\Pr[i\in S]/\Pr[i\in S']\leq \exp\left(\frac{\varepsilon}{\sigma}|p_i'-p_i|\right)$.
By similar argument, we can also prove that $\Pr[i\not\in S]/\Pr[i\not\in S']\leq \exp\left(\frac{\varepsilon}{\sigma}|p_i'-p_i|\right)$.
Hence, $\forall X\subseteq[n]$, $\Pr[S=X]/\Pr[S'=X]\leq \exp\left(\frac{\varepsilon}{\sigma}\|p'-p\|_1\right)\leq \exp(\varepsilon)$ that concludes the proof.
\end{proof}

\begin{lemma}[Sparisity of \textsc{DPSparsification}$(p,\sigma,\varepsilon,\gamma)$]\label{lem:sparsity_dp_sparsifiction}
If $\|p\|_1\leq 1$ and $\gamma\geq \frac{3\sigma}{\varepsilon}\ln(n)$, with probability at least $1-1/n$, the output $S$ of \textsc{DPSparsification}$(p,\sigma,\varepsilon,\gamma)$ satisfies (1) $|S|\leq 3/\gamma$, (2) $\forall i \in S, p_i\geq \gamma/3$, (3) $\forall i$ with $p_i\geq 2\gamma$, $i\in S$.
\end{lemma}
\begin{proof}
If $p_i<\gamma/3$, since $\gamma \geq \frac{3\sigma}{\varepsilon}\ln(n)$, the probability that $i$ is added to $S$ is at most $\frac{1}{2}\exp(-2\ln n)=1/(2n^2)$.
By taking union bound over all $i\in[n]$, with probability at least $1/(2n)$, $\forall i$ with $p_i<\gamma/3$, $i\not\in S$.
Since $\|p\|_1\leq 1$, $|S|\leq 3/\gamma$.
If $p_i\geq2\gamma$, the probability that $i$ is added to $S$ is at least $1-1/(2n^2)$.
By taking union bound over all $i\in[n]$, with probability at least $1/(2n)$, $\forall i$ with $p_i\geq 2\gamma$, $i\in S$.
\end{proof}

%Due to space limits, we put the analysis of theoretical guarantees of our sparsification (Algorithm~\ref{alg:sparsification}) in~\refappendix{sec:proof_of_dp_sparse_ppr}.
As a side result, we obtain DP sparse approximate PPR vector by applying composition (Theorem~\ref{thm:adv_comp}) of the sparsification (Algorithm~\ref{alg:sparsification}) and Laplace mechanism (Theorem~\ref{thm:laplace_mechanism}) on our senstivity-bounded PPR vector (Algorithm~\ref{alg:capped_push_flow}).
We refer readers to~\refappendix{sec:proof_of_dp_sparse_ppr} for more details.
\begin{algorithm}[h]
	\small
	\begin{algorithmic}[1]\caption{\textsc{DPSparsification}$(p,\sigma,\varepsilon,\gamma)$}\label{alg:sparsification}
	\STATE \textbf{Input:} An (approximate PPR) vector $p\in\mathbb{R}^n$, a sensitivity upper bound $\sigma$ of $p$, a parameter $\varepsilon$ for DP, and a threshold $\gamma$.
	\STATE \textbf{Output:} An $\varepsilon$-differentially private set of indices $S\subseteq [n]$.
	\STATE Initialize $S\gets \emptyset$.
	\STATE For each $i\in[n]$, if $p_i\leq \gamma$, add $i$ into $S$ with probability $\frac{1}{2}\cdot \exp\left(-\frac{\varepsilon}{\sigma}\cdot (\gamma-p_i)\right)$, otherwise add $i$ into $S$ with probability $1-\frac{1}{2}\cdot\exp\left(\frac{\varepsilon}{\sigma}\cdot(\gamma-p_i)\right).$
	\STATE Output $S$.
	\end{algorithmic}
\end{algorithm}
In Algorithm~\ref{alg:dp_instantembedding_via_sparse_ppr}, we show how to get DP InstantEmbedding via DP sparse approximate PPR vector.
\begin{algorithm}[h]
	\small
	\begin{algorithmic}[1]\caption{\textsc{DPEmbeddingSparse}$(G,s,\alpha,\xi,\sigma,k,\varepsilon,\tp)$}\label{alg:dp_instantembedding_via_sparse_ppr}
	\STATE \textbf{Input:} Graph $G=(V,E)$, source $s\in V$, teleport probability $\alpha$, precision $\xi$, sensitivity parameter $\sigma$, embedding dimension $k$, DP parameter $\varepsilon$, and $\tp\in \{\joint,\njoint\}$ indicating whether joint $\varepsilon$-DP or $\varepsilon$-DP is required.
	\STATE \textbf{Output:} $\varepsilon$-DP $k$-dimensional embedding vector.
	\STATE $\varepsilon_0\gets \varepsilon/2$.
    \STATE $\hat{p}\gets \textsc{PushFlowCap}(G,s,\alpha,\xi,\sigma,\tp)$.
    \STATE $S\gets \textsc{DPSparsification}\left(\hat{p},\sigma,\varepsilon_0,\frac{3\sigma}{\varepsilon_0}\ln n\right)$.
    \STATE Construct $p$ such that $p_i=\hat{p}_i$ for $i\in S$ and $p_i=0$ for $i\in[n]\setminus S$.
    \STATE $w\gets \textsc{InstantEmbedding}(p,k)$.
    \STATE Draw $Y_1,Y_2,\cdots,Y_k$ independently from $\lap(|S|\cdot \log\left(1+\sigma\cdot n/|S|\right)/\varepsilon_0)$.
    \STATE Output $w+(Y_1,Y_2,\cdots,Y_k)$.
	\end{algorithmic}
\end{algorithm}
We record the theoretical guarantees of the algorithm in Theorem~\ref{thm:sparse_dp_embedding}. %whose proof is deferred to \refappendix{sec:dp_embedding_via_sparse_ppr}.
\begin{theorem}\label{thm:sparse_dp_embedding}
The family of (personalized) algorithms $\{\mathcal{A}_s(G):=\textsc{DPEmbeddingSparse}(G,s,\alpha,\xi,\sigma,k,\varepsilon,\joint)\mid s\in V\}$ is joint $\varepsilon$-DP, and \textsc{DPEmbeddingSparse}$(G,s,\alpha,\xi,\sigma,k,\varepsilon,\njoint)$ is $\varepsilon$-DP with respect to $G$ for any $s\in V$.
In addition, if the input graph $G$ has a minimum degree at least $D$, then the joint $\varepsilon$-DP (resp. $\varepsilon$-DP) output is a $(\xi,O_{\alpha,\varepsilon}(\sigma\log(n)))$-approximate PPR for $\p(s)$ with $O_{\alpha,\varepsilon}\left(\frac{1}{\sigma\log n}\right)$ non-zero entries with probability at least $1-O(1/n)$ when $\sigma\geq \Omega_{\alpha}(1/D^2)$ (resp. $\sigma\geq \Omega_{\alpha}(1/D)$). 
\end{theorem}

\begin{proof}%[Proof of Theorem~\ref{thm:sparse_dp_embedding}]
Firstly, let us consider the DP guarantee. 
According to Lemma~\ref{lem:eps_dp_sparsification}, $S$ is $\varepsilon_0$-DP.
Since $p$ has $|S|$ non-zero entries, the sensitivity of $w$ is at most $|S|\cdot\log(1+\sigma\cdot n/|S|)$.
Due to Laplace mechanism (Theorem~\ref{thm:laplace_mechanism}), given $S$, the final output is $\varepsilon_0$-DP.
Since $S$ is also $\varepsilon_0$-DP, according to composition theorem (Theorem~\ref{thm:adv_comp}), the overall algorithm is $(\varepsilon_0+\varepsilon_0)$-DP, i.e., $\varepsilon$-DP.

Due to the proof of Theorem~\ref{thm:dp_sparse_ppr} (see Appendix~\ref{sec:proof_of_dp_sparse_ppr}), with probability at least $1-O(1/n)$, $p$ is a $(\xi,O(\sigma\log(n)/\varepsilon))$-approximate PPR vector for $\p(s)$, and $|S|\leq O\left(\frac{\varepsilon}{\sigma\log n}\right)$.
Due to Fact~\ref{fac:laplace_bound} and union bound, with probability at least $1-O(1/n)$, $\max_{i\in[k]} |Y_i|\leq |S|\cdot \log\left(1+\sigma\cdot n/|S|\right)/\varepsilon_0\cdot \log n\leq O(\sigma^{-1}\log(1+\sigma\cdot n))$.
\end{proof}

\begin{comment}
\paragraph{DP InstantEmbedding via sparse approximate PPR.}
%According to Theorem~\ref{thm:sensitivity_instantembedding}, when the approximate PPR vector $p$ is sparser, we might get smaller sensitivity of $w$.
%Follow this observation, 
According to Theorem~\ref{thm:sensitivity_instantembedding}, we can obtain lower sensitivity for InstantEmbedding if the approximate PPR vector is sparser.
We first compute an approximate PPR vector $\hat{p}$ with bounded sensitivity by \textsc{PushFlowCap} (Algorithm~\ref{alg:capped_push_flow}), and use \textsc{DPSparsification} (Algorithm~\ref{alg:sparsification}) to find the set of indices $S$ of potentially large entries of $\hat{p}$ in a differently private manner. 
Let $p$ satisfy $\forall v\in S, p_v=\hat{p}_v$ and $\forall v\not\in S,p_v=0$.
We show that the joint sensitivity and the sensitivity of the embedding $w$ is at most $O_{\alpha,\varepsilon}(D^2\log n)$ and $O_{\alpha,\varepsilon}(D\log n)$ respectively when we use the above $p$ in Algorithm~\ref{alg:instant_embedding}.
By applying Laplace mechanism (Theorem~\ref{thm:laplace_mechanism}), we obtain an $\varepsilon$-DP and a joint $\varepsilon$-DP algorithm for InstantEmbedding.
We refer readers to~\refappendix{sec:dp_embedding_via_sparse_ppr} for more details.
\end{comment}

%\section{Proof of Theorem~\ref{thm:dp_sparse_ppr}}
\section{Differetially Private Sparse Approximate PPR}\label{sec:proof_of_dp_sparse_ppr}

\begin{theorem}\label{thm:dp_sparse_ppr}
Given source $s$, teleport probability $\alpha$, precision $\xi$, sensitivity bound $\sigma$ and DP parameter $\varepsilon$, there is an algorithm which is always $\varepsilon$-DP with respect to the input $n$-node graph $G$ and in addition outputs $\left(\xi, O_{\alpha,\varepsilon}(\sigma\ln n)\right)$-approximate PPR vector with $O_{\alpha,\varepsilon}(1/(\sigma\ln n))$ non-zero entries for $\p(s)$with probability at least $1-O(1/n)$ when $G$ has minimum degree at least $\Omega_{\alpha}(1/\sigma)$.
There is a family of (personalized) algorithms $\{\mathcal{A}_1,\mathcal{A}_2,\cdots,\mathcal{A}_n\}$ which is joint $\varepsilon$-DP with respect to the input $n$-node graph $G$ and in addition $\forall s\in V$, $\mathcal{A}_s(G)$ outputs $\left(\xi, O_{\alpha,\varepsilon}(\sigma\ln n)\right)$-approximate PPR vector with $O_{\alpha,\varepsilon}(1/(\sigma\ln n))$ non-zero entries for $\p(s)$ when $G$ has minimum degree at least $\Omega_{\alpha}(\sqrt{1/\sigma})$ with probability at least $1-O(1/n)$.
\end{theorem}

The algorithm that outputs the $\varepsilon$-DP sparse approximate PPR is given in Algorithm~\ref{alg:dp_sparse_ppr}.
\begin{algorithm}[h]
	\small
	\begin{algorithmic}[1]\caption{\textsc{DPSparsePPR}$(G,s,\alpha,\xi,\sigma,\varepsilon,\joint/\njoint)$}\label{alg:dp_sparse_ppr}
	\STATE \textbf{Input:} Graph $G=(V,E)$, source $s\in V$, teleport probability $\alpha$, precision $\xi$, sensitivity bound $\sigma$, DP parameter $\varepsilon$, and $\tp\in \{\joint,\njoint\}$ indicating whether joint $\varepsilon$-DP or $\varepsilon$-DP is considered.
	\STATE \textbf{Output:} $\varepsilon$-DP or joint $\varepsilon$-DP approximate PPR vector for $\p(s)$.
	\STATE $\varepsilon_0\gets \varepsilon/2$.
    \STATE $\hat{p}\gets \textsc{PushFlowCap}(G,s,\alpha,\xi,\sigma,\tp)$.
    \STATE $S\gets \textsc{DPSparsification}\left(\hat{p},\sigma,\varepsilon_0,\frac{3\sigma}{\varepsilon_0}\ln n\right)$.
    \STATE For each $i\in S$, let $Y_i$ be drawn independently from $\lap\left(\sigma/\varepsilon_0\right)$
    \STATE Construct $p$ such that $p_i=\hat{p}_i$ for $i\in S$. Let $p_i=Y_i=0$ for $i\in[n]\setminus S$.
    \STATE Output $p+(Y_0,Y_1,\cdots,Y_n)$.
	\end{algorithmic}
\end{algorithm}
\begin{proof}[Proof of Theorem~\ref{thm:dp_sparse_ppr}]
According to Theorem~\ref{thm:sensitivity_capped_pushflow}, the (joint) sensitivity of $\hat{p}$ is $\sigma$.
According to Lemma~\ref{lem:eps_dp_sparsification}, set $S$ is $\varepsilon_0$-DP.
For any fixed set $S$, the sensitivity of $p$ is always bounded by the sensitivity of $\hat{p}$.
Therefore, given $S$, $p+(Y_0,Y_1,\cdots,Y_n)$ is $\varepsilon_0$-DP due to Laplace mechanism (Theorem~\ref{thm:laplace_mechanism}).
According to the composition theorem, Theorem~\ref{thm:adv_comp}, the final output is $(\varepsilon_0+\varepsilon_0)$-DP which is $\varepsilon$-DP.
According to Lemma~\ref{lem:sparsity_dp_sparsifiction}, with probability at least $1-1/n$, the number of non-zero entries of the output is $O(\varepsilon/(\sigma\ln n))$.

Next, let us consider the accuracy of the output. 
Suppose joint $\varepsilon$-DP is considered.
If $G$ has minimum degree at least $\Omega_{\alpha}(\sqrt{1/\sigma})$,
according to Lemma~\ref{lem:correctness_cap_ppr}, $\hat{p}$ is a $\xi$-approximate PPR vector for $\p(s)$.
Suppose $\varepsilon$-DP is considered.
If $G$ has minimum degree at least $\Omega_{\alpha}(1/\sigma)$,
according to Lemma~\ref{lem:correctness_cap_ppr}, $\hat{p}$ is a $\xi$-approximate PPR vector for $\p(s)$.

Notice that $\gamma=\frac{3\sigma}{\varepsilon_0}\cdot \ln n$.
According to Lemma~\ref{lem:sparsity_dp_sparsifiction}, $\|p-\hat{p}\|_{\infty}\leq O(\gamma)$.
According to Fact~\ref{fac:laplace_bound}, with probability at least $1-1/n$, $\max_{i\in[n]} |Y_i|\leq O(\gamma)$.
Thus, with probability at least $1-O(1/n)$, the output is a $(\xi,O(\gamma))$-approximate PPR vector for $\p(s)$.
\end{proof}

%\section{Differentially Private InstantEmbedding via Sparse Approximate PPR}
%\section{Proof of Theorem~\ref{thm:sparse_dp_embedding}}
%\label{sec:dp_embedding_via_sparse_ppr}

\section{Additional Experimental Results}
\label{app:exp}
\paragraph{Additional baselines} We also considered DPNE~\cite{xu2018dpne}  in our evaluation as a potential baseline. However, the DPNE algorithm as described ~\cite{xu2018dpne} is not DP in the setting of edge-DP or joint-DP unless certain assumptions on the input graph are made.\footnote{This has been confirmed in a personal communication with the authors.} %We contacted the authors of~\cite{xu2018dpne} and confirmed that their algorithm 
%may not provide DP as in edge-DP or joint-edge DP on \emph{arbitrary} inputs like our paper. 
Since the paper does not provide edge-DP or joint-edge DP on \emph{arbitrary} inputs, like our paper, we omit it from our empirical evaluation.

%Following the authors guidance, we attempted to modify their algorithm to implement a version of DPNE that is at joint-DP and thus comparable to our joint-DP algorithm. Unfortunately we were not able to obtain any performance not significantly worse than than the randomized response baseline. For this %Given that our work (which is a valid DP algorithm without further assumptions) cannot be fairly compared to the algorithm in~\cite{xu2018dpne} nor their paper is easily adaptable to our settings we decided to not report such results. 

\paragraph{Details on the datasets}
We provide more detailed description of the datasets.
PPI is a protein-protein interaction dataset, where labels represent hallmark gene sets of specific biological states. Blogcatalog is a social networks of bloggers, where labels are self-identified topics of their blogs. Flickr is a photo-based social network, where labels represent self-identified interests of users and edges represent messages between users.

\iffalse 
\paragraph{Additional studies on parameter settings.}
%\input{figures/approx-appendix}
We report in Figure~\ref{fig:heatplots} an extended version of our analysis of the accuracy of DP PPR in approximating the PPR rankings. We are especially interested in a reliable mechanism for setting the parameter $\sigma$ of our algorithm. We present the result according to three different metrics:  $L_1$ similarity ($1-L_1$ distance), Recall@100, and Spearman rank correlation coefficient $\rho$. The most indicative metric of the three is Recall@100, since it evaluates the quality of the nearest neighborhood of the node. We observe good performance across wide range of $\sigma$.

We report in Figure~\ref{fig:ppr-residual} comparative analysis of the total residual in the joint DP and non-joint DP versions. We observe that for almost all range of the sensitivity parameter $\sigma$, joint DP offers more push compared to non-joint DP version of the algorithm.

\begin{figure*}
\includegraphics[width=\linewidth]{figures/heat_pos.pdf}
\includegraphics[width=\linewidth]{figures/heat_blog.pdf}
\includegraphics[width=\linewidth]{figures/heat_flickr.pdf}
\caption{Sensitivity to the $\sigma$ parameter measured across three different datasets and three different metrics -- $L_1$ similarity, Recall@100, and Spearman rank correlation coefficient $\rho$. Best viewed in colour.\label{fig:heatplots}}
\end{figure*}

\input{figures/residual}

\fi 
{
\color{red}

}